\documentclass[runningheads]{amsart}

\usepackage{amsmath,amssymb,amsthm}
\usepackage{xspace}
\usepackage[disable]{todonotes}
\usepackage{virginialake}
\usepackage[colorlinks]{hyperref}
\usepackage[capitalise]{cleveref}
\usetikzlibrary{arrows}

\begin{document}

\renewcommand{\phi}{\varphi}
\renewcommand{\emptyset}{\varnothing}
\renewcommand{\epsilon}{\varepsilon}
\newcommand{\defname}[1]{\textbf{{#1}}}
\newcommand{\df}{:=}
\newcommand{\bnf}{::=}
\newcommand{\Pow}{\mathcal P}
\newcommand{\pow}[1]{\Pow(#1)}

\newcommand{\IH}{\mathit{IH}}

\newcommand{\Nat}{\mathbb{N}}
\newcommand{\Int}{\mathbb{Z}}
\newcommand{\Cal}[1]{\mathcal{#1}}
\newcommand{\SF}[1]{\mathsf{#1}}
\newcommand{\resp}{\textit{resp.}\xspace}
\newcommand{\aka}{\textit{aka}\xspace}
\newcommand{\viz}{\textit{viz.}\xspace}
\newcommand{\ie}{\textit{i.e.}\xspace}
\newcommand{\cf}{\textit{cf.}\xspace}
\newcommand{\wrt}{\textit{wrt}\xspace}
\newcommand{\rk}[1]{\SF{rk}(#1)}

\newcommand{\mgt}[1]{\textcolor{magenta}{#1}}
\newcommand{\cyn}[1]{\textcolor{cyan}{#1}}
\newcommand{\olv}[1]{\textcolor{olive}{#1}}
\newcommand{\orange}[1]{{\color{orange}#1}}
\newcommand{\purple}[1]{{\color{purple}#1}}

\newcommand{\anupam}[1]{\todo{Anu: #1}}
\newcommand{\abhishek}[1]{\todo{Abhi: #1}}
\newcommand{\note}[1]{\mgt{#1}}


\newtheorem{theorem}{Theorem}
\newtheorem{lemma}[theorem]{Lemma}
\newtheorem{proposition}[theorem]{Proposition}
\newtheorem{observation}[theorem]{Observation}
\crefformat{observation}{Observation~#2#1#3}

\newtheorem{corollary}[theorem]{Corollary}
\newtheorem{fact}[theorem]{Fact}

\theoremstyle{definition}
\newtheorem{definition}[theorem]{Definition}
\newtheorem{example}[theorem]{Example}
\newtheorem{remark}[theorem]{Remark}
\newtheorem{convention}[theorem]{Convention}

\newcommand{\Alphabet}{\mathcal{A}}
\newcommand{\Var}{\mathcal{V}}

\newcommand{\proves}{\vdash}

\newcommand{\Lang}{\mathcal L}
\newcommand{\lang}[1]{\Lang(#1)}

\newcommand{\wLang}{\Lang}
\newcommand{\wlang}[1]{\wLang(#1)}

\newcommand{\FV}{\mathrm{FV}}
\newcommand{\fv}[1]{\FV(#1)}

\newcommand{\fint}[1]{\lceil #1 \rceil}

\newcommand{\FL}{\mathrm{FL}}
\newcommand{\fl}[1]{\FL(#1)}
\newcommand{\eqfl}{=_\FL}
\newcommand{\lefl}{<_\FL}
\newcommand{\leqfl}{\leq_\FL}
\newcommand{\geqfl}{\geq_\FL}

\newcommand{\subform}{\sqsubseteq}
\newcommand{\supform}{\sqsupseteq}
\newcommand{\dle}{\prec}
\newcommand{\dleq}{\preceq}
\newcommand{\dge}{\succ}
\newcommand{\dgeq}{\succeq}

\newcommand{\struct}{\mathfrak L}
\newcommand{\interp}[2]{#2^{#1}}

\newcommand{\es}[1]{\mathcal{#1}}
\newcommand{\ines}[1]{\mathcal{#1}^\geq}

\newcommand{\esEX}[3]{\es {#1}_{#2}(#3)}
\newcommand{\inesEX}[3]{\ines {#1}_{#2}(#3)}

\newcommand{\canes}[1]{\mathcal E_{#1}}
\newcommand{\meetes}[1]{\mathcal M_{#1}}
\newcommand{\meetvar}[2]{X_{#1 \cap #2}}

\newcommand{\rleq}{\lesssim}
\newcommand{\rgeq}{\gtrsim}
\newcommand{\req}{\approx}


\newcommand{\seqar}{\rightarrow}

\newcommand{\id}{\mathsf{id}}

\newcommand{\K}{\mathsf{k}}
\newcommand{\kk}[1]{\K_{#1}}
\newcommand{\wk}{\rr{\mathsf{w}}}
\newcommand{\cntr}{\mathsf{c}}
\newcommand{\cut}{\mathsf{cut}}

\newcommand{\lr}[1]{#1\text{-}l}
\newcommand{\rr}[1]{#1\text{-}r}
\newcommand{\func}{\SF{func}}

\newcommand{\KA}{\mathsf{KA}}
\newcommand{\lhKA}{\ell\KA}
\newcommand{\HKA}{\mathsf{HKA}}

\newcommand{\RLA}{\ensuremath{\mathsf{RLA}}\xspace}
\newcommand{\RLAhat}{\widehat\RLA}

\newcommand{\nRLA}{\nu\RLA}
\newcommand{\nRLAhat}{\nu\RLAhat}

\newcommand{\LRLAhat}{\mathsf L\RLAhat}
\newcommand{\CRLA}{\ensuremath{\mathsf{CRLA}}\xspace}
\newcommand{\LRAind}{\ensuremath{\mathsf{LRA}^{\SF{ind}}}\xspace}

\newcommand{\nLRLAhat}{\nu\LRLAhat}
\newcommand{\nCRLA}{\nu\CRLA}

\newcommand{\infrule}{\mathsf{r}}


\newcommand{\eloise}{$\exists$l\"{o}ise\xspace}
\newcommand{\abelard}{$\forall$belard\xspace}
\newcommand{\prover}{\mathbf{P}}
\newcommand{\denier}{\mathbf{D}}

\title[A proof theory of right-linear ($\omega$-)grammars]{A proof theory of right-linear ($\omega$-)grammars \\ via cyclic proofs}

\author{Anupam Das \and Abhishek De}

\begin{abstract}
Right-linear (or left-linear) grammars are a well-known class of context-free grammars computing just the regular languages. 
They may naturally be written as expressions with (least) fixed points but with products restricted to letters as left arguments, giving an alternative to the syntax of regular expressions.
In this work we investigate the resulting logical theory of this syntax.
Namely we propose a theory of \emph{right-linear algebras} ($\RLA$) over of this syntax and a cyclic proof system $\CRLA$ for reasoning about them.

We show that $\CRLA$ is sound and complete for the intended model of regular languages.
From here we recover the same completeness result for $\RLA$ by extracting inductive invariants from cyclic proofs.
%
Finally we extend $\CRLA$ by \emph{greatest} fixed points, $\nCRLA$, naturally modelled by languages of $\omega$-words thanks to right-linearity. 
We show a similar soundness and completeness result of (the guarded fragment of) $\nCRLA$ for the model of $\omega$-regular languages, this time requiring game theoretic techniques to handle interleaving of fixed points.
\end{abstract}

\maketitle    

\section{Introduction}
{Regular expressions} are the prototypical syntax for describing regular languages.
\emph{Theories} of regular expressions, such as Kleene algebras ($\KA$) e.g.~\cite{Conway71book,Kleene56,Kozen94:completeness-ka}, give rise to alternative interpretations, not least \emph{relational} models, cf.~\cite{Tarski41:calc-of-rels}, which are used in dynamic logics (e.g.~\cite{sep-logic-dynamic}) and for reasoning about program correctness \cite{Koz97:kat,Koz00:hoare-kat,HoaMolStrWeh09:cka}.
The completeness of $\KA$ for the (equational) theory of regular languages, due to Kozen \cite{Kozen94:completeness-ka} and Krob \cite{Krob91:ka-completeness} independently, is a celebrated result that has led to several extensions and refinements, e.g.~\cite{KozSmi97:kat-completeness,KozSil12:left-handed-completeness,CraLauStr15:omega-regular-algebras,KozSil20:left-handed-completeness}.
More recently the proof theory of $\KA$ has been studied via `hypersequential' calculi \cite{DasPou17:hka}, and this has led to novel proofs for completeness results \cite{DasDouPou18:lka-completeness}.

Another well-known description for regular expressions is via right-linear (or left-linear) grammars, widely studied in e.g.\ parsing~\cite{JK75,parsingbook} and compilers~\cite{compilerbook}.
These are context-free grammars in which each production has RHS either $\epsilon$ or $X$ or $aX $, for $a$ a letter from an alphabet $\Alphabet$ and $X$ a non-terminal.\footnote{In fact right-linear grammars can intuitively be construed as nondeterminstic finite automata.}
Context-free grammars may naturally be written using expressions with \emph{least fixed points}, $\mu$, e.g.~\cite{Leiss92:ka-with-rec,EsikLeiss05:alg-comp-semirings,GraHenKoz13:inf-ax-cf-langs,Leiss16:matrices-over-mu-cont-chom-alg}.
For right-linear grammars, instead of arbitrary products $e\cdot f$, we require the left argument $e$ to be atomic.
In this work we investigate these resulting `right-linear $\mu$-expressions' as a starting point for logical systems reasoning about regular languages.

The lack of general products renders the resulting algebras more general (i.e.\ less structured): while Kleene algebras (and friends) are axiomatised over an underlying (idempotent) semiring with identity, our \emph{right-linear algebras} ($\RLA$) have no multiplicative structure at all. 
Instead we construe each letter $a$ as a homomorphism on the underlying (bounded) semilattice $(0,+)$.
Ultimately this has the effect that $\RLA$ carries more models than $\KA$ (and friends): formally, every (even left-handed) Kleene algebra is a right-linear algebra, but the converse does not hold.
For instance the class of $\omega$-languages $\pow {\Alphabet^\omega}$ forms a right-linear algebra in the expected way, as $\omega$-words are closed under concatenation with letters from the left. 
However it is not hard to see that these homomorphisms $a\in \Alphabet$ are not compatible with any multiplicative structure $(1,\cdot)$ on $\pow {\Alphabet^\omega}$.

More pertinent to the present work, the lack of general products significantly simplifies the resulting proof systems.
Complete cut-free systems for regular expressions, such as $\HKA$ from \cite{DasPou17:hka}, can only be obtained\footnote{Note that there have been several conjectures and failed attempts to design cut-free systems based on `intuitionistic' sequents akin to the Lambek calculus. E.g.\ Jipsen proposed such a system in \cite{Jip04:semirings-res-kls}, for which it has been shown that cut is \emph{not} admissible \cite{Busz06:action-logic-eq-theories}. Wurm has also proposed a cut-free system in \cite{Wurm14:ka-reglang-substructlog}, but counterexamples have been found, e.g.\ \cite[Appendix A]{DasPou17:hka-hal-version}.} through complex `hypersequents' of the form $\Gamma \seqar X$, where $\Gamma$ is a list of expressions (interpreted as a product) and $X$ is a set of such lists (interpreted as a set of products).\footnote{In \cite{DasPou17:hka} the authors argue for the necessity of hypersequential structure by consideration of proof search for the theorems $a^* \leq (aa)^* + a(aa)^*$ and $(a+b)^* \leq a^*(ba^*)^*$.}
In the absence of products, it suffices to consider `co-intuitionistic' sequents of the form $e\seqar \Gamma$, where $e$ is an expression and $\Gamma $ is a set of expressions.
At least one advantage here is the simplification of proof search, with (co-)intuitionistic sequents more amenable to automated reasoning than the bespoke hypersequents of $\HKA$.
At the same time the resulting systems seem to admit a much smoother metalogical treatment, as demonstrated by the development we carry out in this work.

\subsection*{Contribution}
We propose a theory $\RLA$ of expressions built up over $0,+,a,\mu$, where $\mu$ is a binder for writing least fixed points.
$\RLA$ naturally has language and relational models like Kleene algebra, but also enjoys `asymmetric' models like the aforementioned $\omega$-languages.
We design a sequent calculus based on this syntax using co-intuitionistic sequents, and define a notion of \emph{cyclic proof} to model leastness of fixed points, giving the system $\CRLA$.
Our first main contribution is the adequacy of this system for regular languages: $\CRLA \proves e=f$ just if the regular languages computed by $e$ and $f$, $\lang e $ and $\lang f$ respectively, are the same.

More interestingly perhaps, we can use this completeness result to obtain the same for $\RLA$, our second main contribution: $\RLA \proves e=f \iff \lang e = \lang f$. Thus the regular languages are essentially the \emph{free} right-linear algebra.
The difficulty here, a recurring leitmotif of cyclic proof theory, is the extraction of \emph{inductive invariants} from cyclic proofs. 
Our method is inspired by previous analogous results for \emph{left-handed} Kleene algebras ($\lhKA$), namely \cite{KozSil12:left-handed-completeness,KozSil20:left-handed-completeness,DasDouPou18:lka-completeness}.
We first carry out some `bootstrapping' in $\RLA$ for solving systems of right-linear systems of equations, similarly to \cite{KozSil12:left-handed-completeness,KozSil20:left-handed-completeness}.\footnote{In \cite{KozSil12:left-handed-completeness,KozSil20:left-handed-completeness}, Kozen and Silva rather call these systems \emph{left-linear}, but we follow the terminology from formal language theory.}
The invariants we extract from $\CRLA$ compute intersections of languages wrt $\lang\cdot$, like in \cite{DasDouPou18:lka-completeness}, but the relative simplicity of our system $\CRLA$ means the corresponding logical properties and intermediate steps are established much more succinctly.

Finally, leveraging the aforementioned fact that $\RLA$ admits $\omega$-languages as a natural model, we develop an extensions $\nCRLA$ of  $\CRLA$ by \emph{greatest} fixed points, $\nu$, for which $\lang \cdot$ extends to a model of $\omega$-regular languages.
Note that this contrasts with approaches via regular expressions, where various `$\omega$-regular algebras' are, formally speaking, \emph{not} (even left-handed) Kleene algebras \cite{Wagner76:omega-regular-algebras,Cohen00:omega,Struth12:left-omega-algebras,CraLauStr15:omega-regular-algebras}.
Our third main contribution is the soundness and completeness of $\nCRLA$ for the model of $\omega$-regular languages.
Compared to $\CRLA$, the difficulty for metalogical resoning here is to control the interleavings of $\mu$ and $\nu$, both in the soundness argument and for controlling proof search in the completeness argument.
To this end we employ \emph{game theoretic} techniques for characterising word membership, and also for controlling proof search. 
Here priority of fixed points is resolved according to a \emph{parity} condition.\footnote{In fact our `$\mu\nu$-expressions' can, in a sense, be construed as a sort of non-deterministic Rabin $\omega$-automaton model.}
Note that, due to an asymmetry between the join and meet structure of language models with respect to lattice homomorphisms, completeness fails in general (e.g.\ since $\lang \cdot \not\models a\top = \top$); instead we prove it for the \emph{guarded} fragment, which already exhausts all $\omega$-regular languages.

All of our main results are summarised in \cref{fig:summary}.

\begin{figure}[t]
    \centering
    \begin{tikzpicture}
    \draw [fill=teal!20,teal!20] (-2,-2) rectangle (2.5,3.5);
    \draw [fill=orange!20,orange!20] (2.5,-2) rectangle (8,3.5);
        \node (A) at (0,3) {$\RLA$};
        \node (B) at (0,1) {$\CRLA$};
        \node (C) at (5,2) {regular languages};
        \node (D) at (0,-1) {$\nCRLA$};
        \node (E) at (5,-1) {$\omega$-regular languages};
        \draw[->] (A) to[bend left] node[below] {\textcolor{magenta}{\scriptsize\textbf{\cref{thm:rla-soundness}}}} (C);
        \draw[->,gray,dotted] (B) to[bend left] node[below] {\textcolor{gray}{\scriptsize\textbf{\cref{thm:crla-soundness}}}} (C);
        \draw[<-] (B) to[bend right] node[above] {\textcolor{magenta}{\scriptsize\textbf{\cref{thm:crla-completeness}}}} (C);
        \draw[->] (D) to[bend left] node[below] {\textcolor{magenta}{\scriptsize\textbf{\cref{thm:ncrla-soundness}}}} (E);
        \draw[<-] (D) to[bend right] node[above] {\textcolor{magenta}{\scriptsize\textbf{\cref{thm:ncrla-guarded-complete}}}} (E);
        \draw[dashed] (-2,0.2) -- (8,0.2);
        \draw[right hook->] (C) to (E);
        \draw[->] (B) to (A);
        \node[left] at (0,2) {\textcolor{magenta}{\scriptsize\textbf{\cref{lemma:clla-to-rla}}}};
        \draw[right hook->] (B) to (D);
    \node[above right] at (-2,-2) {\textcolor{teal}{\small \textit{Proof systems}}};
    \node[above left] at (8,-2) {\textcolor{orange}{\small \textit{Language models}}};
    \end{tikzpicture}
    \caption{Summary of our main contributions. Each arrow $\rightarrow$ denotes an inclusion of equational theories, over right-linear expressions. The gray arrow marked \cref{thm:crla-soundness} is consequence of the black solid ones.
    }
    \label{fig:summary}
\end{figure}
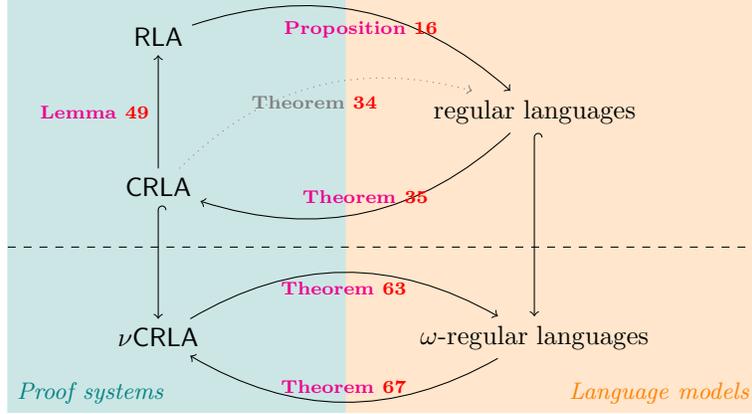

\subsection*{Other related work}
Our syntax of right-linear $\mu$-expressions is essentially the same as Milner's expressions for regular behaviours in the seminal work \cite{Milner84} (modulo inclusion of $1$).
That work gives an axiomatisation for a class of `regular behaviours' over guarded expressions, but this crucially differs from our $\RLA$: in \cite{Milner84} the letters $a\in \Alphabet$ need not satisfy the distributivity law $a(e+f)=ae +af$ (i.e.\ letters are not necessarily $+$-homomorphisms), and so Milner's class of regular behaviours does not form a right-linear algebra in the sense of this work.

The right-linear equational systems we consider when `bootstrapping' $\RLA$ have been widely studied in the context of (left-handed) Kleene algebras, e.g.~\cite{Kozen94:completeness-ka,KozSil12:left-handed-completeness,KozSil20:left-handed-completeness}.
$\RLA$ essentially characterises the bounded semilattices admitting least solutions to such equations, as opposed to $\lhKA$ which characterises the idempotent semirings with this property.

The relationship between $\mu$-expressions and equational systems has been studied in the more general setting of context-free grammars, giving rise to \emph{Park $\mu$-semirings} and \emph{Chomsky algebras}, e.g.\ \cite{Leiss92:ka-with-rec,EsikLeiss05:alg-comp-semirings,GraHenKoz13:inf-ax-cf-langs,Leiss16:matrices-over-mu-cont-chom-alg}.
In this sense our approach of using fixed points to notate right-linear grammars is natural,  
amounting to a restriction of the established syntax for notating context-free grammars.

\todo{should we say something about cyclic proofs?}

Recently \cite{HazKup22:transfin-HKA} has already proposed an adaptation of $\HKA$ reasoning with $\omega$-regular (and certain transfinite) expressions. 
Our system $\nCRLA$ may be seen as an alternative to that system, again operating with structurally simpler lines.

\anupam{commented bullet points from intro discussion below}

    





\section{A theory of right-linear expressions}
In this section we introduce a logical theory of right-linear grammars.
Throughout this work, let us fix a set $\Alphabet$ (aka an \defname{alphabet}) of \defname{letters} $a,b$ etc., and a set $\Var$ of \defname{variables} $X,Y $ etc.

\subsection{A syntax for right-linear grammars}
Recall that a context-free grammar is \defname{right-linear} if each production has RHS (WLoG) either $\epsilon$ or $Y$ or $a Y$ for $Y$ a non-terminal and $a \in \Alphabet$.
It is well-known that right-linear grammars generate \emph{just} the regular languages:

\begin{remark}
[Right-linear grammars as NFAs]
\label{rem:rlgs-as-nfas}
A right-linear grammar is morally just a non-deterministic finite automaton (NFA) by construing non-terminals as states and:
\begin{itemize}
    \item For each production $X\to \epsilon$, $X$ is a final state.
    \item Each production $X\to Y$ is an $\epsilon$-transition from $X$ to $Y$.
    \item Each production $X\to aY$ is an $a$-transition from $X$ to $Y$.
\end{itemize}
In this way we may freely speak of `runs' of a right-linear grammar and use other automaton-theoretic terminology accordingly.
\abhishek{mention start state}
\anupam{do we need to? initialising at a state $X$ just gives the language of that non-terminal.}
\end{remark}

We shall employ a standard naming convention for context-free grammars using `$\mu$-expressions', cf.~\cite{Leiss92:ka-with-rec,EsikLeiss05:alg-comp-semirings,GraHenKoz13:inf-ax-cf-langs,Leiss16:matrices-over-mu-cont-chom-alg}, for which right-linearity corresponds to an expected restriction on products.

\defname{(Right-linear) ($\mu$-)expressions}, written $e,f,\dots$, are generated by:
\[
e,f,\dots \quad ::= \quad 
0 \quad \mid \quad 
1 \quad \mid \quad 
 X \quad \mid \quad
e + f \quad \mid \quad a\cdot e \quad \mid \quad \mu X e
\]
for $a\in\Alphabet$ and $X \in \Var$. We usually just write $ae$ instead of $a\cdot e$. 
We sometimes say ($\mu$-)\defname{formula} instead of `($\mu$)-expression', and write $\subform$ for the subformula relation.

The set of \defname{free variables} of an expression $e$, written $\fv e$, is defined by:
\begin{itemize}
    \item $\fv 0 \df \emptyset$
    \item $\fv 1 \df \emptyset$
    \item $\fv X \df \{X\}$
    \item $\fv {e+f} \df \fv e \cup \fv f$
    \item $\fv {ae} \df \fv e$
    \item $\fv {\mu X e} \df \fv e \setminus \{X\}$
\end{itemize}
We say $e $ is \defname{closed} if $\fv e = \emptyset$ (otherwise it is \defname{open}).

The intended semantics of right-linear expressions is given by {(regular) languages}.
\begin{definition}
    [Regular language model $\lang \cdot$]
    Let us temporarily expand the syntax of expressions to include each language $A\subseteq \Alphabet^*$ as a constant symbol.
We interpret each closed expression (of this expanded language) as a subset of $ \Alphabet^*$ as follows:
\begin{itemize}
    \item $\lang 0 \df \emptyset$
    \item $\lang 1 \df \{\epsilon\}$
    \item $\lang A \df A$
    \item $\lang {e+f} \df \lang e \cup \lang f$
    \item $\lang {ae} \df \{a \sigma : \sigma \in \lang e\}$
    \item $\lang {\mu X e(X)} \df \bigcap \{ A \supseteq \lang {e(A)} \} $
\end{itemize}
\end{definition}

\anupam{comment about monotonicity of operators commented here, as we have formal proof later. Use this in conference version.}

\begin{example}
[Empty and universal languages]
\label{ex:empty-univ-as-rlexps}
    In the semantics above, note that the empty language is also computed by $\mu X X $, as well as $0$, and so in a sense $0$ is syntactic sugar. 
    Including it explicitly in our syntax will nonetheless be helpful when defining structures/models of our syntax.
    Dually the universal language is given by $\mu X \left(1 + \sum\limits_{a \in \Alphabet} a X\right)$ in $\lang\cdot$.
\end{example}

\begin{example}
    [Right-linear grammars]
    \label{ex:rlgram-even-odd-to-rlexps}
    As mentioned previously, our expressions can be seen as a naming convention for right-linear grammars, in particular by fixing an order for resolving non-terminals.
    We shall see a self-contained proof that our expressions compute just the regular languages shortly, but let us first see an example.
    Consider the following (right-linear) grammar:
    \[
    \begin{array}{r@{\quad \rightarrow \quad }l}
    E &  \epsilon \quad \mid \quad aO \\
    O & aE
    \end{array}
    \]
    Clearly $E$ generates just the even length words and $O$ just the odd length words (over singleton alphabet $\{a\}$).
    We can write the solutions of these grammars (wrt.~$\lang\cdot$) as expressions in two ways, left or right below, depending on whether we choose to resolve $E$ or $O$ first:
    \[
    \begin{array}{r@{\ = \ }l}
         O'(X) & aX \\
         E & \mu X (1 + a O'(X)) \\
         O & O'(E)
    \end{array}
    \qquad
    \begin{array}{r@{\ = \ }l}
         E'(X) & 1 + aX\\
         O & \mu X aE'(X) \\
         E & E'(O) 
    \end{array}
    \]
\end{example}

\subsection{Regularity and Fischer-Ladner}
As our expressions are designed to be a syntax for right-linear grammars, we duly have the following result:
\begin{proposition}\label{prop:rl-langs-are-just-reg-langs}
    $L\subseteq \Alphabet^*$ is regular $\iff$ $L = \lang e$ for some $\mu$-expression $e$.
\end{proposition}
One half of this result is readily obtained by an interpretation of regular expressions within our right-linear syntax:
\begin{proof}
    [Proof of $\implies$ direction of \cref{prop:rl-langs-are-just-reg-langs}]
    We give a translation from regular expressions to right-linear expressions preserving the language computed.
    For this we first define right-linear expressions $e\bullet g$, computing the product of $e$ and $g$, where $e$ is a regular expression and $g$ is a right-linear expression.
    We proceed by induction on the structure of $e$:
\begin{equation}
\label{eq:prod-regexp-rlexp-to-rlexp}
    \begin{array}{r@{\ \df \ }l}
        0 \bullet g & 0 \\
         1 \bullet g & g \\
         a \bullet g & ag \\
         (e+f)\bullet g & e\bullet g + f\bullet g \\
         (ef)\bullet g & e \bullet (f\bullet g)  \\
         e^* \bullet g & \mu X (g + e \bullet X )
    \end{array}
\end{equation}
    %
From here we readily interpret each regular expression $e$ simply as $e\bullet 1$.
\end{proof}

\begin{remark}
    [One-variable fragment]
    Note that \eqref{eq:prod-regexp-rlexp-to-rlexp} remains in the `one-variable' fragment of right-linear expressions: we need only one (bound) variable, here $X$. 
\end{remark}

For the other direction, we first recall a standard generalisation of `subformula' or `subexpression' when working with fixed points:

\begin{definition}[Fisher-Ladner]
\label{def:fl}
The \defname{Fischer-Ladner ($\FL$) closure} of an expression $e$, written $\fl e$, is the smallest set of expressions such that:

\begin{itemize}
    \item $e\in\fl e$;
    \item if $f+g \in \fl e$ then $f\in\fl e$ and $g\in\fl e$;
    \item if $af \in \fl e$ then $f\in \fl e$; and
    \item if $\mu Xf(X) \in \fl e$ then $f(\mu X f(X)) \in \fl e$.
\end{itemize}
For a set of expressions $E$ we write $\fl E \df \bigcup\limits_{e\in E}\fl e$.
We write $e \leqfl f$ if $e \in \fl f$, 
$e \lefl f$ if $e\leqfl f \not\leqfl e$ and $e  \eqfl f$ if $e\leqfl f \leqfl e$.
A \defname{trace} is a sequence $e_0 \geqfl e_1 \geqfl \cdots$.
\end{definition}

We have the following well-known properties of the $\FL$ closure:
\begin{proposition}
[Properties of $\FL$]
\label{prop:fl-props}
    We have the following:
    \begin{enumerate}
        \item\label{item:fl-finite} $\fl e$ is finite, and in fact has size linear in that of $e$.
        \item\label{item:fl-preorder} $\leqfl$ is a preorder and $\lefl$ is well-founded.
        \item\label{item:traces-have-least-inf-occ-elem} Every trace has a smallest infinitely occurring element, under $\subform$. 
        Unless the trace is eventually stable, the smallest element is a fixed point formula.
    \end{enumerate}
\end{proposition}
\begin{proof}
    [Proof idea]
    For \ref{item:fl-finite}, note that we have $\fl{\mu X e}= \{f[\mu X e /X] : f \in\fl{e} \}$.
    \ref{item:fl-preorder} is immediate from definitions. For \ref{item:traces-have-least-inf-occ-elem} note that $\leqfl \subseteq (\subform \cup \supform)^*$, whence the property reduces to a more general property about trees. 
\end{proof}



Now we can show that the image of $\lang \cdot$ indeed contains only regular languages:

\begin{proof}
    [Proof of $\impliedby$ direction of \cref{prop:rl-langs-are-just-reg-langs}]
\anupam{NFA presentation commented above.}
To show that each expression $e$ computes a regular language, we construct a right-linear grammar with non-terminals $X_f$ for each $f \in \fl e $ by:
\begin{equation}
\label{eq:rlexp-to-rlgram}
    \begin{array}{r@{\quad \rightarrow \quad }l}
        X_0 & \\
         X_1 & \epsilon \\
         X_{f+g} & X_f \quad \mid \quad X_g \\
         X_{af} & a X_f \\
         X_{\mu X f(X)} & X_{f(\mu X f(X))} 
    \end{array}
    \qedhere
\end{equation}
\end{proof}

The grammar in \eqref{eq:rlexp-to-rlgram} above is essentially what we later call the \emph{canonical system of equations} for right-linear expressions. 
Indeed the fact the least solutions of this grammar coincides with $e$ is not particular to our intended semantics: it will be a consequence of our upcoming axiomatisation.

\subsection{Axiomatisation and models}
We now turn to axiomatising an (in)equational theory of right-linear expressions.
We consider equations between expressions, $e=f$, and write $e \leq f $ for $ e+f = f$.

\begin{definition}
[Axiomatisation $\RLA$]
\label{def:axioms-rla}
    Write $\RLA$ for the following axiomatisation of right-linear expressions:
    \begin{enumerate}
    \item\label{item:join-semilattice} $(0,+)$ forms a bounded (join-)semilattice:
    \[
    \begin{array}{r@{\ = \ }l}
        e + 0 & e \\
         e + f & f + e \\
         e + (f + g) & (e + f ) + g
    \\
         e + e & e
    \end{array}
    \]
    \item\label{item:semilattice-homomorphism} Each $a\cdot $, for $a \in \Alphabet$, is a (bounded) semilattice homomorphism:
    \[
     \begin{array}{r@{\ = \ }l}
     a0 & 0 \\
     a(e+f) & ae + af 
    \end{array}
    \]
    \item $\mu X e(X)$ is a least pre-fixed point of the operator $X \mapsto e(X)$:
    \begin{itemize}
        \item\label{item:prefix-mu} (Prefix) $e(\mu X e(X)) \leq \mu X e(X)$.
        \item\label{item:induction-mu} (Induction) $e (f) \leq f \implies \mu X e(X) \leq f$.
    \end{itemize}
\end{enumerate}
\end{definition}

Note that, immediately from these axioms, we inherit natural order-theoretic axioms from the defined inequality $\leq$:
\begin{itemize}
    \item (Reflexivity) $e\leq e$
    \item (Antisymmetry) $e \leq f$ \& $f \leq e$ $\implies$ $e=f$
    \item (Transitivity) $e\leq f$ \& $f\leq g$ $\implies$ $e \leq g$

    \smallskip

    \item (Boundedness) $0\leq e$
    \item ($+$-intro) $e_i \leq e_0 + e_1$, for $i=0,1$.
    \item ($+$-elim) $e_0 \leq f$ \& $e_1 \leq f$ $\implies$ $e_0 + e_1 \leq f$.

    \smallskip

    \item (Monotonicity) $e\leq f$ $\implies$ $ae \leq af$.
\end{itemize}

\begin{remark}
    [1]
    \label{rem:1-is-just-a-thing}
    Note that we have no axiom involving $1$.
    While $\lang\cdot$ interprets it as the empty word, it is actually consistent with $\RLA$ to interpret it as an arbitrary element.
    This is in contrast to theories of regular expressions like Kleene algebras, where the product operation induces a bona fide multiplicative structure for which $1$ is the identity.
    In right-linear algebras, $1$ is just a `thing'.
\end{remark}

\begin{remark}
    [$0$]
    \label{ex:alt-defs-0}
    Recalling \cref{ex:empty-univ-as-rlexps}, $\RLA$ indeed proves $ \mu X X = 0$.
    The $\geq$ direction is immediate from Boundedness, and for the $\leq $ direction we have:
    \begin{itemize}
        \item $0\leq 0$ by Reflexivity (or Boundedness).
        \item Thus $\mu XX \le 0$ by the Induction axiom. 
    \end{itemize}
    We can carry out similar reasoning to prove $\mu X (aX) = 0$, relying on the fact that $a$ is a \emph{bounded} semilattice homomorphism, $a0=0$.
\end{remark}

The axioms of $\RLA$ suffice to establish monotonicity of all the basic operations:
\begin{proposition}[Functoriality]
\label{prop:functoriality-rla}
$\RLA + \vec f \leq \vec g \proves  e(\vec f) \leq e(\vec g)$
\end{proposition}
\begin{proof}
By induction on the structure of $e(\vec X)$, working in $\RLA$:
\begin{itemize}
    \item If $e(\vec X)$ is $0$ then we have $0\leq 0$ by reflexivity.
    \item If $e(\vec X)$ is $1$ then we have $1\leq 1$ by reflexivity.
    \item If $e(\vec X)$ is $X_i$ then we have $f_i \leq g_i$ by assumption.
    \item If $e(\vec X)$ is $e_0(\vec X) + e_1 (\vec X)$ then we have $e_0(\vec f) \leq e_0 (\vec g) $ and $e_1 (\vec f) \leq e_1 (\vec g)$ by inductive hypothesis, and so indeed $e(\vec f) \leq e(\vec g)$ by $+$-intro and $+$-elim. 
    \item If $e(\vec X)$ is $ae'(\vec X)$ then $e'(\vec f) \leq e'(\vec g)$ by inductive hypothesis, and so $e(\vec f)\leq e(\vec g)$ by monotonicity of $a\cdot$.
    \item If $e(\vec X)$ is $\mu X e'(X,\vec X)$ then by inductive hypothesis we have $e'(f,\vec f) \leq e'(g, \vec g)$ whenever $f\leq g$, so:
    \[
    \begin{array}{r@{\ \leq\ }ll}
    \mu X e(X, \vec g) & \mu X e(X, \vec g) & \text{by reflexivity of $\leq$}\\
         e'(\mu X e(X,\vec g),\vec f) & e'(\mu X e(X,\vec g), \vec g) & \text{by inductive hypothesis} \\
         & \mu X e' (X,\vec g) & \text{by Prefix} \\
         \mu X e'(X,\vec f) & \mu X e'(X,\vec g) & \text{by Induction} \qedhere
    \end{array}
    \]
\end{itemize}
\end{proof}

\begin{example}
    [Postfix]
    \label{ex:postfix-mu}
    By a standard argument mimicking the proof of the Knaster-Tarski theorem, $\RLA$ proves that each formula $\mu X e(X)$ is a postfixed point:
    \begin{itemize}
        \item\label{item:postfix-mu} (Postfix) $\mu X e(X) \leq e (\mu X e(X))$
    \end{itemize}
To see this, by Induction it suffices to show that $e(\mu Xe(X)) $ is a prefixed point, i.e.\ $e(e(\mu X)) \leq e(\mu X e(X))$.
Now by the functors of \cref{prop:functoriality-rla} above it suffices to show $e(\mu Xe(X))\leq \mu X e(X)$, which is just the Prefix axiom.
    So, together with the Prefix axiom, $\mu X e(X)$ is provably a fixed point of $e(\cdot)$ in $\RLA$. 
\end{example}

One of the main results of this work is the coincidence of $\RLA$ and $\lang\cdot$:
\begin{theorem}
[Adequacy of $\RLA$]
\label{thm:soundness-completeness-rla}
    $\RLA \proves e\leq f$ $\iff $ $\lang e \subseteq \lang f$.
\end{theorem}

The soundness direction, $\implies$, is verified by checking each axiom of $\RLA$. 
For the fixed point axioms, let us first establish the following lemma:\anupam{I really don't think we need to give this in so much detail. It's immediate from the definitions that each $A\mapsto \lang{e(A)}$ is monotone.}
\begin{lemma}[Semantic monotonicity]
\label{lemma:semantic-functoriality}
If $\vec A \subseteq \vec B$ then $\lang{e(\vec A)}\subseteq\lang{e(\vec B)}$
\end{lemma}

\begin{proof}
By induction on the structure of $e(\vec X)$. 
\begin{itemize}
    \item If $e(\vec X) = 0$ then the statement is trivial.
    \item If $e(\vec X) = X_i$ then the statement is trivial.
    \item If $e(\vec X)=e_0(\vec X)+e_1(\vec X)$ then:
\begin{align*}
\lang{e(\vec A)} &= \lang{e_0(\vec A)}\cup \lang{e_1(\vec A)} \\
&\subseteq \lang{e_0(\vec B)}\cup \lang{e_1(\vec B)} && \text{by inductive hypothesis}\\
&\subseteq \lang{e(\vec B)}
\end{align*}
\item If $e(\vec X)=ae'(\vec X)$ then:
\begin{align*}
\lang{e(\vec A)} &= a\lang{e'(\vec A)}\\
&\subseteq a\lang{e'(\vec B)} && \text{by inductive hypothesis}\\
&\subseteq \lang{e(\vec B)}
\end{align*}
\item If $e(\vec X)=\mu Y e'(\vec X,Y)$ then:
\begin{align*}
\lang{e(\vec A)} &= \bigcap\{C \supseteq \lang{e'(\vec A,C)}\}\\
&\subseteq \bigcap\{C \supseteq \lang{e'(\vec B,C)}\} && \text{by inductive hypothesis}\\
&\subseteq \lang{e(\vec B)} \qedhere
\end{align*}
\end{itemize}
\end{proof}

\begin{proposition}[Soundness of $\RLA$]
\label{thm:rla-soundness}
    $\RLA \proves e\leq f$ $\implies$ $\lang e \subseteq \lang f$.
\end{proposition}

\begin{proof}
[Proof sketch]
Each axiom other than the fixed point axioms are easily checked. For the fixed point axioms, 
by \cref{lemma:semantic-functoriality} above we have that each $A \mapsto \lang{e(A)}$ is a monotone operation on $\pow{\Alphabet^*}$, and so $\lang {\mu Xe(X)}$ is indeed its least fixed point by the Knaster-Tarski fixed point theorem.\anupam{commented more detailed argument below, reproving KT.}
\end{proof}

The converse, completeness, is more difficult and will be established later in \cref{sec:crla-to-rla}.

\anupam{I commented formal models of $\RLA$ here.}
While this work is mainly proof theoretic, we may sometimes consider models of $\RLA$. 
Note that, as $\mu$ is rather a binder at the level of the language and not a bona-fide operation, 
structures must give an interpretation of all expressions, satisfying the appropriate axioms. 
Explicit definitions of such structures are given for the more general \emph{(Park) $\mu$-semirings} in, e.g.~\cite{EsikLeiss05:alg-comp-semirings,Leiss16:matrices-over-mu-cont-chom-alg}.

\begin{example}
    [Non-standard language interpretations]
    \label{ex:non-standard-interps}
    Recalling again \cref{ex:empty-univ-as-rlexps}, note that $\top' \df \mu X \left(1 + \sum\limits_{a \in \Alphabet} aX\right)$ is not provably a top element in $\RLA$. 
    For this simply consider the structure of all languages over $\Alphabet$, with standard interpretations of $1,+,\mu$, and interpret each $a \in \Alphabet$ as the homomorphism $\emptyset \mapsto \emptyset $ and $A\mapsto \{\epsilon\}$ when $A\neq \emptyset$.
    In this case we actually have $\top' = 1$.\footnote{Note that axiomatising $\top$ elements over Kleene algebras, with respect to language and relation models, is nontrivial \cite{PousWagemaker22}.}
\end{example}

\begin{example}
    [$\omega$-languages]
    \label{ex:omega-langs-are-a-rla}
    $\pow {\Alphabet^\omega}$ forms a model of $\RLA$ under the usual interpretation of $+,a,\mu$, and interpreting $1$ arbitrarily cf.~\cref{rem:1-is-just-a-thing}. 
    Note here that $\omega$-words are closed under concatenation with letters (or even finite words) on the left,
    and least fixed points exist as we are in a powerset lattice, which is complete.
    Referencing the previous \cref{ex:non-standard-interps}, note that this model again does not have $\top'\df \mu X \left(1 + \sum\limits_{a \in \Alphabet} aX\right)$ as a top element.
    E.g.\ if $1$ is interpreted as some $\Box^\omega$, assuming $\Box\in \Alphabet$, then $\top'$ denotes just the $\omega$-words with finitely many occurrences of non-$\Box$ letters (i.e.\ of \emph{finite support}, wrt $\Box$).
\end{example}

\begin{example}
    [Relational models]
    Like Kleene algebras, relational structures also form models of $\RLA$.
    Here the domain is the set of binary relations over a fixed set, and each $a\cdot$ denotes pre-composition by some fixed binary relation. 
    $0$ is the empty relation, $+$ is the union of relations, and $\mu$ gives the least fixed point as expected.\anupam{maybe too vague, be more explicit?}
\end{example}

\subsection{An agebraic viewpoint via systems of equations}
Another way of speaking about models of $\RLA$ is via algebraic structures admitting a certain fixed point theorem.
Indeed this is the role played by \emph{Chomsky algebras} for $\mu$-expressions with general products, cf.~e.g.~\cite{GraHenKoz13:inf-ax-cf-langs,Leiss16:matrices-over-mu-cont-chom-alg}.
The intended semantics there exhausts all context-free languages, but it is nonetheless useful for us to sketch a similar development here.

A \defname{(right-linear) system (of equations)} over variables $\vec X = X_0, \dots, X_{k-1}$ is a set $\es E(\vec X)$ of equations the form $\{X_i = I_i(\vec X)\}_{i<k}$, where each $I_i(\vec X)$ (the \defname{invariant} of $X_i$) is a ($\mu$-free) expression over $\vec X$ (possibly using parameters from a structure).
Write $\ines E (\vec X)$ for $\es E(\vec X)$ with $\geq $ instead of $=$ in each clause.

\anupam{point out similarity to right-linear grammars here}

\begin{definition}
A \defname{right-linear algebra} is a structure $\struct = (L,0,1,+,\cdot )$ where:
\begin{itemize}
    \item $(L,0,+)$ is a bounded (join-)semilattice.
    \item $\cdot : \Alphabet \times L \to L$ where each $a\cdot $ is a bounded semilattice homomorphism.
    \item Each system $\es E (\vec X) = \{X_i = I_i(\vec X)\}_{i<k}$ admits \emph{least solutions} $\vec S \in L$, i.e.\ such that:
    \begin{itemize}
        \item $\struct \models \es E (\vec S)$; and,
        \item $\struct \models \es E (\vec S') \implies \struct \models \vec S \leq \vec S'$.\footnote{Here we are writing $\vec S\leq \vec S'$ for the conjunction of all inequalities $S_i\leq S_i'$.}
    \end{itemize}
\end{itemize}
\end{definition}


In the context-free setting, Chomsky algebras give us an algebraic viewpoint of (standard) Park $\mu$-semirings \cite{GraHenKoz13:inf-ax-cf-langs}. 
We have a similar correspondence here:
\begin{proposition}
\label{thm:rlas-are-just-rla-models}
    The models of $\RLA$ are just the right-linear algebras.
\end{proposition}
The proof is standard, cf.~e.g.~\cite{GraHenKoz13:inf-ax-cf-langs}, but let us briefly explain some of the intermediate steps as we shall make use of some of these concepts later.

\begin{definition}
    [Canonical systems]
    \label{def:canonical-systems}
    Let $E$ be a finite set of expressions closed under $\FL$, i.e.\ $\fl E = E$, and write
    $\vec X_{E} = (X_e)_{e\in E}$.
    The \defname{canonical} system for $E$ is $\es E_{E}(\vec X_{E})$ given by all equations of the form:
\begin{itemize}
    \item $X_0 = 0$
    \item $X_1 = 1$
    \item $X_{ae} = aX_e$
    \item $X_{e+f} = X_e + X_f$
    \item $X_{\mu X e(X)} = X_{e(\mu X e(X))}$
\end{itemize}
\end{definition}


From here we can interpret any right-linear $\mu$-expression within right-linear algebras simply as least solutions to their canonical systems.
%
%
The other direction of \cref{thm:rlas-are-just-rla-models} boils down to what is known as \emph{Beki\'c's theorem}, reducing least solutions of equational systems to individual least fixed points.
The argument is similar to our reasoning in \cref{ex:rlgram-even-odd-to-rlexps}, and the same result holds in our right-linear syntax:
\begin{lemma}
    [Essentially Beki\'c \cite{Bekic84}]
    \label{lem:bekic}
     Working in a $\RLA$ model $\struct$, consider a `system' $\es E(X,Y)$ given by,
     \[
     \begin{array}{r@{\ = \ }l}
          X & e(X,Y) \\
          Y & f(X,Y)
     \end{array}
     \]
     where $e(X,Y)$ and $f(X,Y)$ may be \emph{arbitrary} expressions (with $\mu$s) using parameters from $\struct$.
     There are least solutions $E,F$ in $\struct$ given by,
     \[
     \begin{array}{r@{\ \df \ }l}
          E & \mu X (e(X,f'(X)) \\
          F & f'(E)
     \end{array}
     \]
     where $f'(X) \df \mu Y f(X,Y)$.
\end{lemma}
Repeatedly applying the above lemma allows us to resolve systems of equations within any model of $\RLA$, and hence any model of $\RLA$ is a right-linear algebra. 



\subsection{Comparison to (left-handed) Kleene algebras}
\label{sec:lhKAs}
A \defname{left-handed Kleene algebra} (lhKA) is a structure $\mathfrak K = (K,0,1,+,\cdot, *)$ where $(K,0,1,+,\cdot)$ is an idempotent semiring (aka a \emph{dioid}) and:\footnote{A \defname{Kleene algebra} further satisfies symmetric properties for the $*$, namely $A^* = 1 + A^* A$ and $AB \leq A \implies AB^* \leq B$.}

   \begin{itemize}
    \item $A^* = 1 + AA^*$
    \item $AB\leq B \implies A^*B\leq B$
\end{itemize} 
It is known that lhKAs also admit least solutions to (right-linear) equational systems, e.g.\ \cite{KozSil12:left-handed-completeness,KozSil20:left-handed-completeness}.\footnote{The authors of \cite{KozSil12:left-handed-completeness,KozSil20:left-handed-completeness} refer to such systems as `left-linear'.} 
By restricting multiplication to only letters in the first argument we thus have:
\begin{proposition}
\label{prop:lkas-are-rlas}
    Each left-handed Kleene algebra is also a right-linear algebra.
\end{proposition}

However the converse does \emph{not} hold:
\begin{example}
    [RLAs not lhKAs]
    Recalling \cref{rem:1-is-just-a-thing}, note that we can redefine the interpretation of $1$ in $\lang\cdot$ to be an arbitrary language, so that it is not longer an identity for products.
More generally, the model $\pow{\Alphabet^\omega}$ of \cref{ex:omega-langs-are-a-rla} does not extend to \emph{any} Kleene algebra.
For this, suppose that the homomorphisms $a \in \Alphabet$ do indeed extend to a product operation $\cdot$, satisfying $(aA)\cdot B = a(A\cdot B)$.
This forces the product to just be the left-projection, and so $\cdot$ cannot have any right-identity.
\end{example}

So while it is known that the equational theory of left-handed 
 Kleene algebras is complete for the model of regular languages, this does not immediately entail an analogous result for $\RLA$.
Indeed \cref{thm:soundness-completeness-rla} seems somewhat stronger, given \cref{thm:rlas-are-just-rla-models,prop:lkas-are-rlas}. 
Nonetheless, our proof will use methods similar to those that have been employed for left-handed Kleene algebra in \cite{DasDouPou18:lka-completeness}.

\section{A cyclic system for right-linear algebra}
\label{sec:crla}
In this section we introduce a sequent-based \emph{cyclic} proof system $\CRLA$ for right-linear $\mu$-expressions, and show its soundness and completeness for the regular language model $\lang\cdot$.
Introducing $\CRLA$ here serves two purposes: (a) we use its completeness to infer the completeness of the axiomatisation $\RLA$ later in \cref{sec:crla-to-rla}; and (b) this section is a `warm-up' for the more complex development we require later for treating $\omega$-words via \emph{greatest} fixed points in \cref{sec:gfp}.

\subsection{A sequent calculus}
A \defname{sequent} is an expression $e \seqar \Gamma$, where $\Gamma$ is a set of expressions (called a \defname{cedent}).\abhishek{Isn't it more standard to say $\Gamma$ is a succedent and $e$ is the antecent and they are both cedents?}\anupam{$e$ is not a cedent, it is a formula. I prefer to just say LHS and RHS, for accessibility.}
We interpret the RHSs of sequents as sums of their elements, and duly write $\lang \Gamma \df \bigcup\limits_{f\in\Gamma}\lang f$.
As usual we omit braces when writing cedents, and use commas for set-union.

\begin{remark}
    [On the data structure of sequents]
    The structure of sequents we require for our right-linear syntax is significantly simpler than those of previous systems for reasoning about regular languages, \cite{DasPou17:hka,DasDouPou18:lka-completeness}. 
    Those works, working directly with regular expressions which are closed under arbitrary products, require a system $\mathsf{HKA}$ working with lines of the form $\Gamma \seqar S$, where $\Gamma $ is a list of expressions and $S$ is a (multi)set of lists of expressions, comprising a `hypersequential' structure.
    Instead the current notion of sequent is essentially \emph{co-intuitionistic}: rather than having only one formula on the RHS, we have only one on the LHS.
\end{remark}

\begin{figure}
    \textbf{Identity, modal and structural rules:}
    \[
    \vlinf{\id}{}{e \seqar e}{}
    \qquad
    \vlinf{\kk a}{}{ae \seqar a\Gamma}{e\seqar \Gamma}
    \qquad
    \vlinf{\wk}{}{e \seqar \Gamma, f}{e \seqar \Gamma}
    \]

    \medskip
    \textbf{Left logical rules:}
    \[
    \vlinf{\lr 0}{}{0 \seqar \Gamma}{}
    \qquad
    \vliinf{\lr +}{}{e + f \seqar \Gamma }{e \seqar \Gamma}{f \seqar \Gamma}
    \qquad
    \vlinf{\lr \mu }{}{\mu X e(X) \seqar \Gamma}{e(\mu X e(X)) \seqar \Gamma}
    \]

    \medskip
    \textbf{Right logical rules:}
    \[
    \vlinf{\rr 0}{}{e \seqar \Gamma,0}{e \seqar \Gamma}
    \qquad
    \vlinf{\rr +}{}{e \seqar \Gamma, f_0+f_1}{e \seqar \Gamma, f_0,f_1}
    \qquad
    \vlinf{\rr \mu}{}{e \seqar \Gamma, \mu X e(X)}{e \seqar \Gamma, e (\mu X e(X))}
    \]
    \caption{Rules of the system $\LRLAhat$.}
    \label{fig:llahat}
\end{figure}

We define the system $\LRLAhat$ by the rules in \cref{fig:llahat}.
Here we are writing $a\Gamma \df \bigcup\limits_{f\in \Gamma} af$ in the $\kk a$ rule.
Note here that we allow $\Gamma $ to be empty too.\anupam{this is where $a0=0$ is being used.}

As usual for left or right steps the \defname{principal} formula is the distinguished formula on the left-hand side (LHS) or right-hand side (RHS), respectively, of the lower sequent, as typeset in \cref{fig:llahat}.
Any distinguished formulas on the LHS or RHS, respectively, of the upper sequent are called \defname{auxiliary}.

It is not hard to see that each rule of $\LRLAhat$ is sound for $\RLA$:
\begin{lemma}
[Local soundness]
\label{lemma:local-soundness}
     $\LRLAhat\proves e\seqar \Gamma$ $\implies$ $\RLA \proves e \leq \sum \Gamma$.
\end{lemma}
\begin{proof}
    [Proof sketch]\todo{expand this into a proof for each case to make sure the axioms are correct. remember to check empty RHS case for $\K$ rule etc., which needs $a0=0$.}
    The soundness of each non-$\mu$-rule follows from the semilattice axioms of $\RLA$.
    Note that for the $\kk a$ rule we require $a0\leq 0$ for the case when the RHS is empty.
    The soundness of $\lr\mu$ follows from the Prefix axiom \ref{item:prefix-mu}, and the soundness of $\rr\mu$ follows from the Postfix principle \ref{item:postfix-mu}.
\end{proof}

\begin{remark}
    [Hats and cuts]
    \label{rem:hats-and-cuts}
    In fact the soundness argument above goes through in the subtheory $\RLAhat$ of $\RLA$ obtained by replacing the Induction axioms \ref{item:induction-mu} by the Postfix axioms \ref{item:postfix-mu}.
    In this case we also have the converse result
    in the presence of a \emph{cut} rule:
    \begin{equation}
        \label{eq:cut}
        \vliinf{\cut}{}{e \seqar \Gamma, \Delta}{e\seqar \Gamma,f}{f\seqar \Delta}
    \end{equation}
    We do not use this property and, in general, we avoid consideration of the cut rule in order to simplify the definition and metatheory of the cyclic system we are about to define.
\end{remark}

\begin{definition}
[Cyclic system]
    A \defname{preproof} (of $\LRLAhat$) is generated \emph{coinductively} from the rules of $\LRLAhat$.
    I.e.\ it is a possibly infinite tree of sequents (of height $\leq \omega$) generated by the rules of $\LRLAhat$.
    A preproof is \defname{regular} (or \defname{cyclic}) if it has only finitely many distinct sub-preproofs.
    
    A preproof is \defname{progressing} if each infinite branch has infinitely many $\lr \mu$ steps.
    We write $\CRLA$ for the class of regular progressing $\LRLAhat$-preproofs, which we may simply call $\CRLA$-proofs henceforth.
\end{definition}

\begin{remark}
    [Simplicity of progress condition]
    \label{rem:prog-crla-equiv-inf-lhs}
    Note that our progress condition is equivalent to simply requiring that each infinite branch has infinitely many LHS or $\K$ steps, as every such rule besides $\lr \mu $ reduces the size of the LHS, bottom-up.
    Usually in non-wellfounded proof theory a more complex progress condition is employed, at the level of \emph{traces} (or \emph{threads}, in substructural settings). 
    Indeed we shall require such granularity when admitting greatest fixed points in \cref{sec:gfp}, but the simplicity of the current setting accommodates the simple formulation above.
    For the reader familiar with cyclic proof theory note that: (a) with only least fixed points and no negative operators, progress is witnessed by only LHS traces; (b) with only singleton LHS, there is always a unique LHS trace; and (c) LHS traces are progressing just if they are infinitely often principal, in the absence of greatest fixed points.
\end{remark}

\begin{proposition}
    [Functoriality]
    \label{prop:functoriality-crla}
    There are $\CRLA$ derivations of the form,
    \[
    \vliqf{}{}{e(\vec f) \seqar e(\vec g)}{\{ f_i \seqar g_i \}_i}
    \]
    for each expression $e(\vec X)$.
\end{proposition}
\begin{proof}
    By induction on the structure of $e(\vec X)$:
    \begin{itemize}
     \item If $e(\vec X)  = 0$ the end-sequent is an  instance of $\lr 0$.
        \item If $e(\vec X)  = 1$ the end-sequent is an identity.
        \item If $e(\vec X) = a e'(\vec X)$ we apply the IH to $e'(\vec X)$ and then a $\kk a$ step.
        \item If $e(\vec X) = e_0 (\vec X) + e_1 (\vec X)$ we apply the IH to $e_0(\vec X) $ and to $e_1(\vec X)$ and combine the derivations by $\lr +$ (and some weakenings $\wk$).
        \item If $e(\vec X) = \mu X e'(X,\vec X)$ then we construct:
        \[
        \vlderivation{
        \vlin{\lr\mu,\rr\mu}{\bullet}{\mu X e'(X,\vec f) \seqar \mu X e'(X,\vec g)}{
        \vliiq{\IH}{}{e'(\mu X e'(X,\vec f),\vec f) \seqar e' (\mu X e'(X,\vec g),\vec g)}{
            \vlin{\lr\mu,\rr\mu}{\bullet}{\mu X e'(X,\vec f) \seqar \mu X e'(X,\vec g)}{\vlhy{\vdots}}
        }{
            \vlhy{ \{f_i \seqar g_i\}_i }
        }
        }
        }
        \]
        where the lines marked $\IH$ are obtained by the inductive hypothesis on $e'(X,\vec X)$, and $\bullet$ marks roots of identical subproofs.
        Note that there is only one infinite branch, looping on $\bullet$, along which there are indeed infinitely many $\lr \mu$ steps. \qedhere
    \end{itemize}
\end{proof}

Notice that the only identities required in the above proof were on $1$, unlike the analogous result for $\RLA$, \cref{prop:functoriality-rla}, where we (crucially) needed general identities. 
These functors can thus be employed to reduce general identity to atomic identity, 
a leitmotif of non-wellfounded and cyclic proof theory.
\begin{corollary}
    [$\eta$-expansion of identities]
    \label{cor:eta-expansion-identities}
    The restriction of $\CRLA$ to only identities on $1$, i.e.\ with only initial sequent $1 \seqar 1$, proves the same sequents as $\CRLA$.
\end{corollary}

One of the motivational examples from previous work on the proof theory of regular languages \cite{DasPou17:hka} was the parity principle: $a^* \leq (aa)^* + a(aa)^*$: every word has even or odd length. 
This had no regular (cut-free) proof in an intuitionistic-style sequent system $\mathsf{LKA}$ (\cite{DasPou17:hka,DasPous18:lka-pt}), necessitating the hypersequential lines of $\mathsf{HKA}$ proposed in that work.
Let us see how $\CRLA$ handles the same principle, written in our right-linear syntax.

\begin{example}
    [Parity]
    Recalling \cref{ex:rlgram-even-odd-to-rlexps}, let us write:
    \[
    \begin{array}{r@{\ \df \ }l}
         E & \mu X (1 + aaX) \\
         O & a E
    \end{array}
    \]
    Again, $E$ and $O$ compute languages of words of even and odd length, respectively, over the singleton alphabet $\{a\}$, in the regular language model $\lang \cdot$.
    Let us also write $\top_a \df \mu X (1 + aX)$, cf.~\cref{ex:empty-univ-as-rlexps,ex:non-standard-interps}, which denotes the universal language over $\{a\}$.
    Here is a proof that each word has even or odd parity,
    \[
    \vlderivation{
    \vlin{\rr +}{}{\top_a \seqar E+O}{
    \vlin{\lr\mu,\rr\mu}{\bullet}{\top_a \seqar E, O}{
    \vliin{\lr+,\rr+}{}{1 + a \top_a \seqar 1 + aaE,O}{
        \vlin{\id}{}{1\seqar 1}{\vlhy{}}
    }{
        \vlin{\kk a}{}{a\top_a \seqar aaE,O}{
        \vlin{\lr \mu,\rr\mu}{\bullet}{\top_a \seqar O,E}{\vlhy{\vdots}}
        }
    }
    }
    }
    }
    \]
    where lines marked $\bullet$ are roots of identical subproofs.
    There is only one infinite branch, looping on $\bullet$, and indeed there are infinitely many $\lr\mu$ steps along this branch.
    Note that we have omitted some weakening steps in the left-subproof, a convention that we shall typically employ throughout.
\end{example}

\anupam{Does it make sense to include the (translation of) $(a+b)^* \leq a^* (ba^*)^*$?}

\anupam{I commented old soundness above and replaced with argument below directly using previous NFA/RLG construction, avoiding ranks etc. NB: this is morally the same as using evaluation puzzles.}
\subsection{Soundness}
Given that the proofs of $\CRLA$ are not necessarily well-founded, its soundness does not follow immediately from \cref{lemma:local-soundness},
so we better prove it:

\begin{theorem}
    [Soundness]\todo{can strengthen this to progressing $\LRLAhat$ preproofs.}
    \label{thm:crla-soundness}
    If $\CRLA \proves e\seqar \Gamma $ then $ \lang e \subseteq \lang \Gamma$.
\end{theorem}

Soundness of cyclic proofs is usually established by a `contradiction from infinite descent' argument, exploiting the corresponding progress condition. 
This relies on the subtle notion of \emph{signatures} (or \emph{markings} or \emph{assignments}), computing suitable approximants of fixed points in order to make appropriate choices in the construction of a `fallacious' branch. 
In this work we shall manage such choices at the level of evaluation, taking advantage of automaton models and, later, games for deciding problems of the form `$w\in \lang e$?'

\begin{proof}
[Proof of \cref{thm:crla-soundness}]
Let $P$ be a $\CRLA$ proof of $e \seqar \Gamma$ and suppose $w \in \lang e$.
Construing the canonical system for $e$, as in \eqref{eq:rlexp-to-rlgram}, as an NFA, cf.~\cref{rem:rlgs-as-nfas}, let $\vec X = X_{e_0}, X_{e_1}, \dots , X_{e_n}$ be an accepting run of $w$, with $e_0 = e$ and $e_n = 1$. 
Note that, since RHS rules are non-branching, this run induces a unique (maximal) branch $B$ through $P$ whose LHS is always some $e_i$, in particular by choosing the premiss of a $\lr+$ step according to the run.

Now, since the run $\vec X$ is finite, there are only finitely many LHS steps in $B$ and so, by the progress condition, $B$ is finite, say $(e_{i_j} \seqar \Gamma_j)_{j\leq N}$ with:
\begin{enumerate}
    \item\label{item:induced-branch-base-case} $e_{i_0} = e_0 = e$ and $\Gamma_0 = \Gamma$.
    \item\label{item:induced-branch-step-case} Either $e_{i_{j+1}} = e_{i_j}$ (at an RHS step) or $e_{i_{j+1}} = e_{i_j + 1}$ (at an LHS or $\kk a$ step).
    \item\label{item:runs-end-at-init} $e_{i_N} = \Gamma_N $ as $B$ is maximal and so must terminate at an initial sequent $\id$.
\end{enumerate} 

\anupam{maybe there could be an intermediate lemma here about finite branches.}
\abhishek{do we? Can't we just say wlog assume we only have infinite branches because finite branches are taken care of by local soundness?}
\anupam{I don't think the statement of local soundness is refined enough. stet for now.}
Write $w_i$ for the corresponding word being read from state $X_{e_i}$, in particular so that $w_i \in \lang {e_i}$.
We show that $w_{i_j} \in \lang {\Gamma_j}$, by induction on $N-j$:
\begin{itemize}
    \item If $j=N$ then we have $w_{i_N} \in \lang {e_{i_N}}$ and $e_{i_N} = \Gamma_N$ by \ref{item:runs-end-at-init}.
    \item If $e_{i_j} \seqar \Gamma_j$ concludes any step that is not an $a$ step then we have $w_{i_j} = w_{i_{j+1}}$ and $\lang {\Gamma_j} \supseteq \lang{\Gamma_{j+1}}$, and so, by inductive hypothesis, $w_{i_j} \in \lang{\Gamma_{j}}$.
    \item Otherwise at a $\kk a$ step we have $e_{i_j} = ae_{i_{j+1}}$ and $\Gamma_j = a\Gamma_{j+1}$ and $w_{i_{j+1}} = aw_{i_j}$.
    By the inductive hypothesis we have $w_{i_{j+1}} \in \lang{\Gamma_{j+1}}$, and so $w_{i_{j}} = aw_{i_{j+1}} \in \lang {a\Gamma_{{j+1}}} = \lang{\Gamma_{j}}$ as required.
\end{itemize}
Finally, when $N-j = N$ (i.e.\ $j=0$) we have $w = w_{0} = w_{i_0} \in \lang {\Gamma_0} = \lang {\Gamma}$, as required.
\end{proof}

\subsection{Completeness}
The main result of this subsection is the converse of Soundness, \cref{thm:crla-soundness}, of the previous section:

\begin{theorem}[Completeness of $\CRLA$]
\label{thm:crla-completeness}
    If $\lang e \subseteq \lang \Gamma$ then $\CRLA \proves e \seqar \Gamma$.
\end{theorem}

This is actually rather simpler to prove for $\mu$-expressions when they are \emph{guarded}, cf.~\cref{sec:guarded-fragment} later, as every preproof will automatically have infinitely many $\kk a$ steps, and so also infinitely many LHS steps (by finiteness of expressions).
From here we still need an adequate proof search strategy that preserves validity wrt.~$\lang\cdot$, bottom-up.
In guarded modal logics it usually suffices to establish a `Modal Lemma', which in our setting is:

\begin{proposition}
    [Modal]
    \label{prop:modal}
    If $\lang{ae} \subseteq  \lang{a\Gamma, \Delta}$, where $\Delta $ contains only expressions of the form $1$ or $bf$ for $b\neq a$, then $\lang e \subseteq \lang \Gamma$.
\end{proposition}

While this is somewhat obvious in the regular language model $\lang \cdot$, note that it is not in general true in all $\RLA$ models, e.g.\ if $b$ and $a$ are identical homomorphisms or if $a$ is constant.
However for every other rule (except $\wk$) we have a bona fide invertibility property, already in the subtheory $\RLAhat \subseteq \RLA$ from \cref{rem:hats-and-cuts}:

\begin{proposition}
    [Invertibility]
\label{prop:invertibility}
    For each non-$\{\wk,\kk a\}$ $\LRLAhat$ step $\vlinf{}{}{e \seqar \Gamma}{\{e_i \seqar \Gamma_i\}_{i<k}}$, for some $k\leq 2$, we have $\RLA + e\leq \sum \Gamma \proves e_i \leq \sum \Gamma_i $, for each $i<k$.  
\end{proposition}

\abhishek{this is a bit subtle. It suffices to consider RLA proof only if we have cuts, right?}
\anupam{I don't understand: do you mean the statement could be refined?}

Together these two observations allow us to conduct proof search while preserving validity for $\lang \cdot$, bottom-up, avoiding any deadlocks, by inspection of the rules of $\LRLAhat$.
As mentioned before, restricted to guarded expressions this would be enough to conclude completeness, as every preproof would be progressing.
For the general result, as we state it in \cref{thm:crla-completeness}, we require more machinery:

\begin{lemma}[Weakening]
\label{lem:weakening-lemma}
Consider a (necessarily non-branching, finite) derivation with only right logical steps of form,
\begin{equation}
    \label{eq:wk-lem-input}
    \vlderd{P}{}{e \seqar \Gamma, \Lambda}{e \seqar \Gamma, \Lambda}
\end{equation}
where each formula in $\Lambda$ is principal at least once in $P$.
Then $\lang \Lambda \subseteq \lang \Gamma$.
\end{lemma}

\begin{proof}
[Proof sketch]
We exploit the already established soundness of $\CRLA$, \cref{thm:crla-soundness}, constructing (mutually dependent) cyclic proofs $P^f$ for each $f\in \Lambda$ of $f\seqar \Gamma$.

First we define a transformation of the form,
\[
\vlderd{Q}{}{e\seqar \Gamma,\Lambda}{e\seqar \Gamma, \Lambda'}
\qquad \mapsto\qquad
\vltreeder{Q_f}{f \seqar \Gamma}{}{\{f' \seqar \Gamma\}_{f'\in \Lambda'}}{}
\]
 when $Q$ consists of only right logical steps with no formula of $\Gamma$ ever principal, for each $f\in \Lambda$, with $Q_f$ a finite $\LRLAhat$ derivation (with the indicated premisses). 
We proceed by induction on the length of $Q$, essentially converting each right step to its corresponding left step, formally by analysis of the topmost step:
\begin{itemize}
    \item If $Q$ is empty then we set $Q_f \df f \seqar \Gamma$.
    \item If $Q$ extends $Q'$ above by a step $\vlinf{\rr 0}{}{e \seqar \Gamma, \Lambda',0}{e \seqar \Gamma, \Lambda'} $ then we set:
    \[
    Q_f \ \df \quad
    \vlderivation{
    \vliiq{Q_f'}{}{f \seqar \Gamma}{
        \vlhy{ \{f' \seqar \Gamma\}_{f' \in \Lambda'} }
    }{
        \vlin{\lr 0}{}{0 \seqar \Gamma}{\vlhy{}}
    }
    }
    \]
    \item If $Q$ extends $Q'$ above by a step $\vlinf{\rr+}{}{e \seqar \Gamma, \Lambda', g_0 + g_1}{e \seqar \Gamma, \Lambda', g_0 , g_1}$ then we set:
    \[
    Q_f \ \df \quad
    \vlderivation{
    \vliiq{Q'_f}{}{f\seqar \Gamma}{
        \vlhy{ \{f' \seqar \Gamma\}_{f' \in \Lambda'} }
    }{
        \vliin{\lr +}{}{g_0 + g_1 \seqar \Gamma}{
            \vlhy{g_0 \seqar \Gamma}
        }{
            \vlhy{g_1 \seqar \Gamma}
        }
    }
    }
    \]
    If already $g_i\in \Gamma$ then we must further close off the premiss $g_i \seqar \Gamma$ by applying weakenings and an identity.
    \item If $Q$ extends $Q'$ by a step $\vlinf{\rr \mu}{}{f\seqar \Gamma, \Lambda', \mu X g(X)}{f\seqar \Gamma, \Lambda', g(\mu X g(X))}$ then we set:
    \[
    Q_f \ \df \quad
    \vlderivation{
    \vliiq{Q'_f}{}{f \seqar \Gamma}{
        \vlhy{ \{f' \seqar \Gamma\}_{f' \in \Lambda'} }
    }{
        \vlin{\lr \mu}{}{\mu X g(X) \seqar \Gamma}{\vlhy{g(\mu X g(X)) \seqar \Gamma}}
    }
    }
    \]
    Again, if already $g(\mu X g(X)) \in \Gamma$, then we must further close off the premiss $g(\mu X g(X)) \seqar \Gamma$ by applying weakenings and an identity.
\end{itemize}
Now, returning to \eqref{eq:wk-lem-input}, we define our family of $\CRLA$-proofs $P^f$ for each $f\in \Lambda$ of $f\seqar \Gamma$ coinductively by setting:
\[
P^f \ \df \quad
\vliqf{P_f}{}{f \seqar \Gamma}{
\left\{ 
\vltreeder{P^{g}}{g \seqar \Gamma}{\quad }{ }{\quad }
\right\}_{g\in \Lambda}
}
\]
Since each formula of $\Lambda$ is principal at least once in $P$, each derivation $P_f$ contains at least one left rule, by construction.
Thus each $P_f$ has infinitely many LHS steps and so must be progressing, cf.~\cref{rem:prog-crla-equiv-inf-lhs}.
Note also that, since each $P_f$ is finite, each $P^f$ has only finitely many distinct subproofs, and so is indeed regular.
\end{proof}

We are now ready to put together the proof of our main completeness result:

\begin{proof}
[Proof of \cref{thm:crla-completeness}]
We apply the following bottom-up proof strategy, always preserving validity of the current sequent (wrt.~$\lang \cdot$):
\begin{enumerate}
    \item\label{item:apply-left-rules} Apply left logical rules as long as possible.
    This process can only terminate at a sequent $l\seqar \Gamma$ with $l$ of form $1$ or $ae'$.
    \item\label{item:apply-right-rules} Apply right logical rules \emph{fairly}: always choosing the oldest possible formula as the next principal formula.
    Stop when a sequent repeats and each of its RHS formulas not of form $1$ or $bg$ is principal at least once in between. 
    This gives a sequent of form $l \seqar \Delta, \Lambda$, where $\Delta$ consists only of formulas of the form $1$ or $bg$, and $\Lambda$ consists only of $+$ or $\mu $ formulas.
    \item\label{item:weaken-fixed-points} Weaken all of $\Lambda$ to obtain the sequent $l\seqar a \Gamma', \Delta$.
    \item\label{item:weaken-wrong-letters-and-apply-k} 
    If $l=1$ and $1\in \Delta$, weaken $a\Gamma'$ and the rest of $\Delta$ to terminate at the initial sequent $1\seqar 1$.
    Otherwise $l=ae'$ so, writing $\Gamma'\df \{g : ag \in \Delta\}$, weaken the rest of $\Delta$ to obtain the sequent $ae'\seqar a\Gamma'$, then apply a $\kk a $ step to obtain the sequent $e' \seqar \Gamma'$.
    Go back to \ref{item:apply-left-rules} and continue proof search.
\end{enumerate}
\ref{item:apply-left-rules} and \ref{item:apply-right-rules}
preserve validity by \cref{prop:invertibility}. 
Now \ref{item:weaken-fixed-points} preserves validity by \cref{lem:weakening-lemma}, and \ref{item:weaken-wrong-letters-and-apply-k} by \cref{prop:modal}. 
So, by preservation of validity, the proof search strategy can only terminate at an initial sequent $1\seqar 1$ in \ref{item:weaken-wrong-letters-and-apply-k}. 

Note that \ref{item:apply-right-rules} must indeed terminate as there are only finitely many distinct formulas and sequents that may occur during proof search.
So, since \ref{item:apply-right-rules}, \ref{item:weaken-fixed-points} and \ref{item:weaken-wrong-letters-and-apply-k} are finite, any infinite branch produced has infinitely many left steps from \ref{item:apply-left-rules}.
Thus a preproof produced by this strategy is indeed progressing, cf.~\cref{rem:prog-crla-equiv-inf-lhs}. 

Finally note that each sequent $f \seqar \Pi$ can only conclude at most $|\Pi|$ many distinct subproofs: \ref{item:apply-left-rules}, \ref{item:weaken-fixed-points} and \ref{item:weaken-wrong-letters-and-apply-k} are memoryless, 
\abhishek{I understand what you mean by memoryless, will readers?} 
\anupam{good point. Suggestions?}
and in \ref{item:apply-right-rules} the fairness constraint guarantees that each (possible) formula of $\Pi$ is principal after at most $|\Pi|$ repetitions. 
Since there are only finitely many distinct possible sequents during proof search, a proof produced by this strategy is indeed regular.
%
%
%
%
\end{proof}

\todo{give examples and say something about complexity?}

\section{Completeness of \texorpdfstring{$\RLA$}{} via \texorpdfstring{$\CRLA$}{}}
\label{sec:crla-to-rla}

In this section we use the completeness of $\CRLA$ for $\lang\cdot$, \cref{thm:crla-completeness}, as an engine for proving completeness of the theory of right-linear algebras:

\begin{theorem}
    [Completeness of $\RLA$]
    \label{thm:rla-completeness}
    If $\lang e \subseteq \lang f$ then $\RLA \proves e \seqar f$.
\end{theorem}

The first two subsections of this section carry out some `bootstrapping' for $\RLA$, namely some properties of canonical systems and a construction of big meets similar to those for left-handed Kleene algebras (see, e.g., \cite{KozSil12:left-handed-completeness,KozSil20:left-handed-completeness,DasDouPou18:lka-completeness}).
The final two subsections prove certain properties of big meets and give a translation of $\CRLA$ proofs to $\RLA$. 
These are inspired by similar developments in \cite{DasDouPou18:lka-completeness}, but our exposition seems much simpler thanks to (a) the convenience of right-linear syntax; and (b) the bootstrapping we carried out for systems of equations in $\RLA$.

\subsection{Characteristic properties of canonical systems}
\label{sec:char-prop-canonical-sys}

Recall the definition of \emph{canonical system} $\esEX E E {\vec X_E}$ of a finite set $E$ of expressions closed under $\FL$, from \cref{def:canonical-systems}.
Let us write $I_e (\vec X_E)$ for the invariant of each $X_e$, i.e.\ so that $X_e = I_e (\vec X_E)$ is the clause for $X_e$ in $\esEX E E {\vec X_E}$.
We now show that $\fl E $ are in fact provably the least solutions of $\esEX E E {\vec X_E}$ in $\RLA$:
\begin{theorem}
[Least solutions of canonical systems]
\label{thm:solns-can-sys-characterisation}
For $\FL$-closed finite sets $E$:
    \begin{enumerate}
        \item\label{item:E-solves-its-can-sys} $\RLA \proves \esEX E E E $ 
        and,
        \item\label{item:E-is-lpfp-of-its-can-sys} $\RLA + \inesEX E E {\vec X_E} \proves {E} \leq \vec X_E$.
    \end{enumerate}
\end{theorem}

Note that above we are construing the orderings of $E$ and $\vec X_E$ to match up in the obvious way. We are writing $E\leq \vec X_E$ for the conjunction of all $e\leq X_e$, for $e\in E$.
To prove the theorem
we first need the following `reflection' principle:
\begin{lemma}
[Reflection]
    \label{lem:reflection}
    $\RLA + \inesEX E E {\vec X_E} \proves {E} \proves e(\vec X_E) \leq X_{e(E)}$, for each $e(E) \in E$.
\end{lemma}
\begin{proof}
    By induction on the structure of $e(\vec X_E)$, working in $\RLA + \inesEX E E {\vec X_E}$:
    \begin{itemize}
    \item If $e(\vec X_E)=0$ then $X_0$ then $X_0 \geq 0$ by $\inesEX E E {\vec X_E} $ (or Boundedness).
        \item If $e(\vec X_E)=1$ then $X_1 \geq 1$ by $\inesEX E E {\vec X_E} $.
        \item If $e(\vec X_E) = X_f$ then $e(\vec X_E) \leq X_{e(E)}$ is just an identity on $X_f$.
        \item If $e(\vec X_E) = af(\vec X_E)$ then we have:
        \[
        \begin{array}{r@{\ \leq \ }ll}
             f(\vec X_E) & X_{f(E)} & \text{by inductive hypothesis} \\
             af(\vec X_E) & aX_{f(E)} & \text{as $a$ is a homomorphism} \\
             e(\vec X_E) & X_{af(E)} & \text{by dfn.\ of $e(\vec X_E)$ and $\inesEX E E {\vec X_E}$}
        \end{array}
        \]
        \item If $e(\vec X_E) = e_0(\vec X_E) + e_1(\vec X_E)$ then we have:
        \[
        \begin{array}{r@{\ \leq \ }ll}
             e_0(\vec X_E) & X_{e_0(E)} & \text{by inductive hypothesis} \\
             e_1(\vec X_E) & X_{e_1(E)} & \text{by inductive hypothesis} \\
             e_0(\vec X_E)  + e_1(\vec X_E) & X_{e_0(E)} + X_{e_1(E)} & \text{as $+$ forms a semilattice} \\
             e(\vec X_E) & X_{e(E)} & \text{by dfn.\ of $e(\vec X_E)$ and $\inesEX E E {\vec X_E}$}
        \end{array}
        \]
        \item If $e(\vec X_E) = \mu X e'(X ,\vec X_{E'})$, where $E' \df E\setminus\{e(E)\}$, then we have:\footnote{Note here that, $e(E)\in E$ means that $e(E)$ cannot be a proper subformula of $e(E)$ itself, by as formula construction is wellfounded.}
        \[
        \begin{array}{r@{\ \leq \ }ll}
             e'(X_{e(E)}, \vec X_E) & X_{e'(e(E), E')} & \text{by inductive hypothesis} \\
                & X_{e(E)} & \text{by dfn.\ of $e(\vec X_E)$ and $\inesEX E E {\vec X_E}$} \\
            \mu X e'(X,\vec X_{E'}) & X_{e(E)} & \text{by Induction}\\
            e(\vec X_E) & X_{e(E)} & \text{by dfn.\ of $e(\vec X_E)$} \qedhere
        \end{array}
        \]
    \end{itemize}
\end{proof}

From here we can readily prove our earlier characterisation of canonical solutions:

\begin{proof}
    [Proof of \cref{thm:solns-can-sys-characterisation}]
    For \ref{item:E-solves-its-can-sys}, each clause $e = I_e(E)$ of $\esEX E E E $ is just an identity unless $e$ has form $\mu X e'(X)$. 
    In this case the corresponding clause $\mu X e'(X) = e'(\mu X e'(X))$ follows from the fact that $e$ is a fixed point, cf.~\cref{ex:postfix-mu}.\footnote{Note that this argument goes through already in $\RLAhat \subseteq \RLA$, where $\mu X e'(X)$ is axiomatised as a (not necessarily least) fixed point.}

    \ref{item:E-is-lpfp-of-its-can-sys} follows immediately from \cref{lem:reflection} above as $e\in E$ is closed. 
\end{proof}

\subsection{Big meets}
\label{sec:big-meets}
Let us write $m,n,\dots $ for expressions of the form $1$ or $ae$.

Fix a finite set $E$ of expressions closed under $\FL$, i.e.\ $\fl {E} = E$.
%
%
We define a (infinite) system of equations $\meetes {E}(X_{\vec e})_{\vec e \in  E^*}$ consisting of all equations of form:\footnote{To be clear $\vec e$ should not be confused with the product $\prod \vec e$, which recall is not a right-linear expression: $\vec e$ is just a list of expressions for which, at the risk of ambiguity, we omit explicit delimiters for succinctness of notation.}
\begin{itemize}
    \item $X_{\vec n 0 \vec g} = 0$
    \item $X_{\vec n (f_0 + f_1) \vec g} = X_{\vec n f_0 \vec g} + X_{\vec n f_1 \vec g}$
    \item $X_{\vec n \mu X f(X) \vec g} = X_{\vec n f(\mu X f(X)) \vec g}$
    \item $X_{\vec 1} = 1$
    \item $X_{a\vec e} = aX_{\vec e}$
\end{itemize}

\abhishek{I'm a bit confused about the notation here. What does $\vec n 0 \vec g$ mean? It is not a well-formed right-linear expression.}
\anupam{I added a footnote}

Note that, while $\meetes{E}(X_{\vec e})_{\vec  e \in E^*}$ is an infinite set, each variable depends on only finitely many others; moreover, for each $k\in \Nat$, we have that the restriction $\meetes{E}^k(X_{\vec e})_{\vec e \in E^k}$ to only clauses for $X_{\vec e}$ with $\vec e \in E^k$ is a well-defined finite system of equations.
We point out that, when $k=1$, $\meetes E^1 (X_{e})_{e \in E}$ is just (a renaming of) the canonical system $\esEX E E {\vec X_E}$ for $E$, cf.~\cref{def:canonical-systems}.

\todo{change $\bigcap$ to $\bigsqcap$}
Duly let us (suggestively) write $\bigcap \vec e$ for the least solutions of $\meetes{E}(X_{\vec e})_{\vec  e \in E_0^*}$ in $\RLA$ ($E$ will always be clear from context).
Before addressing the semantics of these solutions, let us establish some basic recursive properties for later use:
%

\begin{lemma}
[Recursive properties]
\label{lem:rec-props-of-meets}
    $\RLA$ proves all equations of the form:
    \begin{itemize}
     \item $\bigcap{\vec e 0 \vec g} = 0$
    \item $\bigcap{\vec e (f_0 + f_1) \vec g} = \bigcap{\vec e f_0 \vec g} + \bigcap{\vec e f_1 \vec g}$
    \item $\bigcap{\vec e \mu X f(X) \vec g} = \bigcap{\vec e f(\mu X f(X)) \vec g}$
    \item $\bigcap{\vec 1} = 1$
    \item $\bigcap{a\vec e} = a\bigcap{\vec e}$
    \end{itemize}
\end{lemma}
Note, in the clauses above, that the prefix $\vec e$ may be arbitrary, unlike in the clauses of $\meetes{E}(X_{\vec e})_{\vec  e \in E_0^*}$.
\begin{proof}
    We consider each case separately: 
    \begin{itemize}
        \item $\bigcap{\vec e 0 \vec g} \geq 0$ as $0$ is a bottom element in $\RLA$.
        For the converse, we just set $0$ to be the solution of every $X_{\vec e0\vec g}$ and note that $0 = 0+0$ in $\RLA$, for the $+$ cases, and that the $\vec 1 $ and $a\vec e $ cases do not present.
        \item For $\bigcap{\vec e (f_0 + f_1) \vec g} \geq \bigcap{\vec e f_i \vec g} + \bigcap{\vec e f_1 \vec g}$, it suffices to show that $\bigcap \vec e(f_0+f_1)\vec g$ is a solution for $X_{\vec ef_i \vec g}$ in $\meetes{E}^\geq(X_{\vec e})_{\vec  e \in E^*}$.
        Each clause is immediate except the critical case for $X_{\vec n f_i \vec g} $. For this we simply first notice that $\bigcap \vec n (f_0 + f_1) \vec g \geq \bigcap \vec n f_i \vec g $ by the $\meetes{ E}$-equations, thence simulating the clause for $X_{\vec n f_i \vec g}$.
        The converse is proved similarly, showing that $\bigcap{\vec e f_i \vec g} + \bigcap{\vec e f_1 \vec g}$ is a solution for $X_{\vec e (f_0 + f_1) \vec g}$ in $\meetes{E}^\geq(X_{\vec e})_{\vec  e \in E^*}$.
        For each case we rely on monotonicity of $+$ (on the RHS) and idempotency of $+ $ (on the LHS) in $\RLA$.
        \item $\bigcap{\vec e \mu X f(X) \vec g} \geq \bigcap{\vec e f(\mu X f(X)) \vec g}$ is proved similarly to above, but goes through more simply as we can just identify the solutions to $X_{\vec n f(\mu X f(X)) \vec g}$ and $X_{\vec n \mu X f(X)\vec g}$ in $\meetes {E}^\geq  (X_{\vec e})_{\vec e \in E^*} $.
        \item $\bigcap{\vec 1} = 1$ follows immediately from the $\meetes {E} $ equations.
        \item $\bigcap{a\vec e} = a\bigcap{\vec e}$ follows immediately from the $\meetes {E} $ equations. \qedhere
    \end{itemize}
\end{proof}



\subsection{Intermezzo: the problem of truth conditions for meets}
We can view the solution $\bigcap \vec e$ as a sort of \emph{product construction} on right-linear grammars (equivalently NFAs).
Thus semantically, within the regular language model $\lang \cdot$, we duly have that $\bigcap \vec e$ indeed computes the suggested intersection:
\begin{proposition}
    $\lang {\bigcap \vec e } = \bigcap\limits_{e \in \vec e} \lang {e} $.
\end{proposition}
To exploit this property when reasoning in $\RLA$ we would ideally like meets to satisfy their basic truth conditions in $\RLA$, namely:
\begin{enumerate}
\item\label{item:binary-meet-elimination} $ \bigcap \vec e \leq e$, for each $e \in \vec e$.
    \item\label{item:binary-meet-introduction} $ \bigwedge_i  e \leq f_i \rightarrow e \leq \bigcap \vec f $
\end{enumerate}

In fact the second condition \emph{fails} in some right-linear algebras, so there is no hope of proving it outright. 
For instance consider any model of $\RLA$ satisfying $0\lneq a \lneq b$.
Then we certainly have $a\leq a$ and $a\leq b$, but by the $\meetes E$-equations we have that $a\cap b = 0$, and so certainly we do not have $a \leq a \cap b$.
More reasonable might be to reframe the second condition as an admissible rule:
\begin{enumerate}
\setcounter{enumi}{2}
    \item\label{item:binary-meet-introduction-rule} $ \forall i \, \proves  e\leq f_i \implies \proves e \leq \bigcap \vec f$.
\end{enumerate}
However even this property for $\RLA$, while eventually true a fortiori, seems to evade direct proof: while some structural induction on $\RLA$ derivability seems necessary, the case of the Induction rule renders such argument difficult to prove outright, without resorting to completeness of $\RLA $ itself.

In the absence of these properties, even associativity of binary meets is not readily derivable, which is why we more generally consider big meets.\anupam{this point might be lost here, could use more example running through the paper, or set up this intermezzo differently.}
Instead, we will be able to establish \ref{item:binary-meet-elimination} and a much weaker version of \ref{item:binary-meet-introduction-rule}, namely monotonicity wrt\ an external relation $\rleq$ extracted from our system $\CRLA$.

\subsection{Structural properties of big meets in \texorpdfstring{$\RLA$}{}}
%
Let us write $\Gamma \rleq \Gamma'$ for the smallest preorder on cedents (i.e.\ sets of formulas) satisfying:
\begin{itemize}
\item $\Gamma  \rleq \Gamma,0$
\item $\Gamma  \rleq \Gamma,e$
    \item $\Gamma, e_i \rleq \Gamma, e_0 + e_1$
    \item $\Gamma, e(\mu X e(X)) \rleq \Gamma, \mu X e(X) $.
\end{itemize}

By inspecting the rules of $\LRLAhat$ we have:
\begin{observation}
\label[observation]{obs:rleq-above-below}
    If $e\seqar \Gamma$ occurs below $e'\seqar \Gamma' $ in a $\LRLAhat$ preproof and there are no $\K$-steps in between, then $e' \rleq e$ and $\Gamma \rleq \Gamma'$.
    \end{observation}

As expected, we also have soundness of $\rleq$ for $\RLA$:
\begin{observation}
    $\Gamma \rleq \Gamma' \implies \RLA \proves \sum \Gamma \leq \sum \Gamma'$.
\end{observation}

Now, by the recursive properties of big meets, \cref{lem:rec-props-of-meets}, and induction on the definition of $\rleq$ we duly have:
\begin{proposition}
    [Monotonicity]
    \label{prop:lla-mon-in-rleq}
     If $\Gamma \rleq \Gamma'$ then $\RLA \proves \bigcap \vec e (\sum\Gamma) \vec g \leq \bigcap \vec e (\sum \Gamma')\vec g$.
\end{proposition}

\anupam{i commented remark about ordering sums above, i think the proof makes it clear that the ordering does not matter.}

\begin{proof}
    Reflexivity and transitivity of $\rleq$ reduce to the same properties of entailment in $\RLA$, so we check each initial clause.
    In all cases we note that, by the recursive properties of big meets from \cref{lem:rec-props-of-meets} and properties of $+$ in $\RLA$, it suffices to show  $\RLA \proves \bigcap \vec e f \vec g \leq \sum\limits_{f'\in \Gamma'} \bigcap \vec e f'\vec g$ for each $f\in \Gamma$.
    Each of these reduce to basic properties of $\RLA$ and again the recursive properties from \cref{lem:rec-props-of-meets}.
\end{proof}

The recursive properties also allow us to recover certain structural properties of meets, analogous to some of $+$ in $\RLA$:\anupam{consider placing this proposition earlier and treat $\vec e$ as sets already for monotonicity above.}
\begin{proposition}
    [Structural properties]
    \label{prop:meet-struct-props}
    $\RLA$ proves:
    \begin{enumerate}
    \item (Commutativity) $\bigcap \vec e f_0 f_1 \vec g \leq \bigcap \vec e f_1 f_0 \vec g$.
    \item (Elimination) $\bigcap \vec e f_0 f_1 \vec g \leq \bigcap \vec e f_i \vec g$.
        \item (Idempotency) $\bigcap \vec e f \vec g \leq \bigcap \vec e f f \vec g$
    \end{enumerate}
\end{proposition}
\begin{proof}
We prove each item separately:    
\begin{enumerate}       
\item Each $\bigcap \vec e f_1 f_0 \vec g$ is a solution of $X_{\vec e f_0f_1\vec g}$ by using the recursive properties from \cref{lem:rec-props-of-meets} for each clause.
        \item Each $\bigcap \vec e f_i \vec g$ is a solution of $X_{ \vec e f_0f_1 \vec g}$ by using the recursive properties from \cref{lem:rec-props-of-meets} for each clause, requiring idempotency of $+$ when $f_{1-i}$ is a sum.
        \item Each $\bigcap \vec e f f \vec g$ is a solution of $\bigcap \vec e f \vec g$ by using the recursive properties from \cref{lem:rec-props-of-meets} twice for each clause, requiring some algebraic properties of $+$ when $f$ is a sum. \qedhere
    \end{enumerate}
\end{proof}

Thus we may henceforth safely treat $\vec e$ as sets when considering big meets $\bigcap \vec e$.
In particular, given $E_0$, we shall construe the variables of $\meetes{E_0}$ as being indexed by (now one of finitely many) subsets of $E_0$, e.g.\ writing $\meetes{E_0}(X_E)_{E\subseteq E_0}$.

\subsection{Computing an invariant from \texorpdfstring{$\CRLA$}{} proofs}
\label{sec:inv-of-crla-prfs}

Let us fix a $\CRLA$ proof $P$ of $e_0 \seqar \Gamma_0$.
We assume WLoG that $P$ applies left logical rules maximally before applying other rules and that all identities are restricted to $1\seqar 1$, cf.~the proof search strategy in the proof of $\CRLA$ completeness, \cref{thm:crla-completeness} (see also \cref{cor:eta-expansion-identities}).

We set $E_0 = \fl{\{\sum \Gamma: \Gamma \text{ an RHS in $P$}\}}$ and construe all big meets $\bigcap F$, for $F\subseteq E_0$, as solutions of $X_{F}$ in $\meetes {E_0} {(X_{E})_{E\subseteq E_0}}$.
Note, again, that the canonical system $\esEX E {e_0} {\vec X_{e_0}}$, with $\vec X_{e_0} = (X_e)_{e\in \fl {e_0}}$, is just the restriction of 
$\meetes {E_0} {(X_{E})_{E\subseteq E_0}}$ to variables $X_e$ with $e \in \fl {e_0}$.

Write $G_e$ for the set of $\sum\Gamma $ such that $e\seqar \Gamma$ is a sequent in $P$. 
We will show:

\begin{lemma}
[Invariant]
\label{lem:invariant}
    $\bigcap G_e$, for $e \in \fl {e_0}$, are $\RLA$-solutions for $X_e$ in $\inesEX{E} {e_0} {\vec X_{e_0}}$.
\end{lemma}
\begin{proof}
    We proceed by case analysis on the structure of $e \in \fl {e_0}$.
    \begin{itemize}
        \item $\bigcap G_0 \geq 0$ by Boundedness.
        \item $\bigcap G_1 \geq 1$.
        First note that, for any sequent $1\seqar \Gamma$ in $P$, there must be an initial sequent $1\seqar 1$ above it with no $\K$-steps in between, whence we conclude by~\cref{obs:rleq-above-below} and \cref{prop:lla-mon-in-rleq}.
        \item $\bigcap G_{\mu X e(X))} \geq \bigcap G_{e(\mu X e(X))}$ follows from Observation~\ref{obs:rleq-above-below} and \cref{prop:lla-mon-in-rleq} as $P$ must apply left rules maximally before other rules.
        \item $\bigcap G_{e_0 + e_1} \geq \bigcap G_{e_0} + \bigcap G_{e_1}$ also follows from Observation~\ref{obs:rleq-above-below} and \cref{prop:lla-mon-in-rleq} as $P$ must apply 
        left rules maximally before other rules.
        \item $\bigcap G_{ae} \geq a\bigcap G_e$. 
        Since $P$ only has identities on $1$ and is progressing, so left-principal infinitely often, for every sequent $ae \seqar \Gamma$ there is a (first) $\K$-step above $\vlinf{\kk a}{}{ae \seqar a\Gamma'}{e\seqar \Gamma'}$.
        So we have $\Gamma' \in G_e$ and
        by Observation~\ref{obs:rleq-above-below} we have $\Gamma \rgeq a\Gamma'$. 
        Since the choice of $\Gamma \in G_{ae}$ was arbitrary, we have by \cref{prop:lla-mon-in-rleq} and \cref{prop:meet-struct-props} that $ \bigcap G_{ae} \geq \bigcap\limits_{\Gamma' \in G_e} a\Gamma' \geq  a \bigcap G_e$, by the $\meetes{E_0}^\geq$ equations, as required. \qedhere
    \end{itemize}
\end{proof}

\begin{lemma}
\label{lemma:clla-to-rla}
$\CRLA\proves e\seqar f$ $\implies$ $\RLA\proves e\leq f$
\end{lemma}

\begin{proof}
By \cref{lem:invariant} and \cref{thm:solns-can-sys-characterisation} we have $\RLA \proves e\leq \bigcap G$ for some $G\ni f$, for meets over an appropriate set. \abhishek{typo? Should it be $G_e$ instead of $G$?}
\anupam{no, it says `some G', which is in particular $G_e$ earlier, but we do not need to name it. only the properties stated are used.}
Thus, by Elimination from \cref{prop:meet-struct-props}, we in particular have $e\leq f$, as required.
\end{proof}

\todo{Observe that using this and RLA soundness, we have an alternate proof of CRLA soundness}

From here we are finally able to immediately deduce the main result of this section, the completeness of $\RLA$ for $\lang\cdot$:

\begin{proof}
[Proof of \cref{thm:rla-completeness}]
If $\lang e \subseteq \lang f$
    we have by completeness of $\CRLA$, \cref{thm:crla-completeness}, that $\CRLA \proves e\seqar f$, and thus by~\cref{lemma:clla-to-rla} we have $\RLA\proves e\leq f$.
\end{proof}

Notice that our soundness result for $\CRLA$, \cref{thm:crla-soundness}, is now actually subsumed by \cref{lemma:clla-to-rla,thm:rla-soundness}, cf.~\cref{fig:summary}.
\section{Greatest fixed points and \texorpdfstring{$\omega$}{}-languages}
\label{sec:gfp}

\todo[inline]{Consider reworking this section (and next) simply without 1, so that we have only $\omega$-length words.}

In this section we extend the cyclic framework we developed in \cref{sec:crla} to a syntax with \emph{greatest} fixed points.
These naturally compute $\leq \omega$-length words; in particular, the right-linear syntax means that $\omega$-words are closed under all operations of the syntax, in particular left-concatenation with letters.

The metatheory of the resulting system, with both least and greatest fixed points, is much more complex than for $\CRLA$, so we split it into two sections.
The main result of this section is the soundness of our system.

\subsection{Extension by greatest fixed points}
We extend the grammar of expressions by:
\[
e,f,\dots \quad \bnf \quad \dots \quad \mid \quad \nu X e
\]
We call such expressions \defname{$\mu\nu$-expressions}, if we need to distinguish them from expressions without $\nu$.
The intention is that $\nu$ behaves like a \emph{greatest} fixed point, governed by rules dual to the $\mu$ stating that it is a greatest post-fixed point:
\begin{enumerate}
    \item (Postfix) $\nu X e(X) \leq e(\nu X e(X))$
    \item (Coinduction) $f \leq e(f) \implies f\leq \nu Xe(X)$
\end{enumerate}
While this work will not focus on axiomatisations of $\mu\nu$-expressions, let us call the extension of $\RLA$ by the axioms above $\nu\RLA$.



    


\begin{remark}
[Finite words]
Note that languages of finite words already comprises a complete lattice and so we may readily interpret $\mu\nu$-expressions within them in the expected way.
%
    However note that $\nu\RLA$ is \emph{not} complete for this model. In particular the inclusion $\nu X (1 + aX)\leq \mu X (1 + a X) $ cannot have any proof, as it has a countermodel by way of the $\omega$-word semantics we are about to give.
\end{remark}

\abhishek{good place to mention that up until now, all this development could be done with `left-linear' expressions. In the following *right*-linearity is crucial.}
\anupam{I think this is not worth saying: it is an artefact of the fact we are thinking of $\omega$-words as being written left-right rather than right-left (which would equally be a fine convention).}

\begin{definition}
    [Intended semantics of $\mu\nu$-expressions]
    We consider languages of finite and $\omega $-words, i.e.\ subsets of $\Alphabet^{\leq \omega}$. 
    We extend the definition of $\lang \cdot $ to $\mu\nu$-expressions by setting:
    \begin{itemize}
        \item $\wlang {\nu X e(X)} \df \bigcup \{ A \subseteq \wlang {e(A)} \}$
    \end{itemize}
\end{definition}
It is worth pointing out at this juncture that $\pow{\Alphabet^{\leq \omega}}$ indeed forms a complete lattice under $\subseteq$, and that the right-linear syntax means that it is closed under concatenation with letters on the left.
Since all the operations are monotone, $\wlang{\nu X e(X)}$ is indeed the greatest fixed point of the operation $A\mapsto \lang{e(A)}$, by the Knaster-Tarski theorem.

\begin{remark}
    [Finite words as $\omega$-words of finite support]
    \label{rem:fin-words-as-omega-words}
    Usually $\omega$-languages are construed as containing \emph{only} $\omega$-words, and not any finite words. 
    Indeed we can accommodate this interpretation by setting $\wlang 1 \df \Box^\omega$ (for $\Box \in \Alphabet$ arbitrary), so that finite words are construed as `finitely supported' $\omega$-words, cf.~\cref{ex:omega-langs-are-a-rla}.
    We do not take this route so as to remain consistent with our previous exposition, formally extending the definition of the model $\lang \cdot$. 
    However let us point out that, in the presence of greatest fixed points, we could also simply omit $1$ from our syntax altogether: it is not difficult to see that, for any $1$-free expression $e$, we have $\lang e \subseteq \Alphabet^\omega$.
\end{remark}

\begin{remark}
    [$\top$, a subtlety]
    \label{rem:top-nu-morphisms}
    In contrast with the $\mu$-only setting, cf.~\cref{ex:non-standard-interps}, we can express a bona fide top element by $\top \df \nu X X$. 
    Clearly $\wlang \top = \Alphabet^{\leq \omega}$ contains every other language, this is indeed provable in $\nRLA$:
    \begin{itemize}
        \item $e \leq e$ by Reflexivity.
        \item $e \leq \nu X X = \top$ by Coinduction.
    \end{itemize}
    However notice that we do \emph{not} require letters $a\in \Alphabet$ to be interpreted as homomorphisms that preserve the top element.
    Indeed, in $\wlang \cdot$, this does not hold, as $a\Alphabet^{\leq \omega} \neq \Alphabet^{\leq \omega}$. 
    When it comes to our proof system $\nCRLA$ later, this possibility is once again accommodated by the asymmetry of sequents: while the RHS can be empty, the LHS must be a single expression.
\end{remark}

A language is \defname{$\omega$-regular} just if it is computed by an expression of the form $\sum\limits_{i<n} e_i f_i^\omega$, where $e_i,f_i$ are regular expressions with $\epsilon \notin \lang {f_i}$ for each $i<n$.
It is not hard to see that all $\omega$-regular languages are computed by $\mu\nu$-expressions:

\begin{proposition}
    [Adequacy]
    \label{prop:omega-reg-have-munu-exps}
    $\mu\nu$-expressions exhaust all $\omega$-regular languages.
\end{proposition}
\begin{proof}
    Recall the definition of the concatenation $eg$ of regular expressions $e$ and right-linear expressions $g$ from \eqref{eq:prod-regexp-rlexp-to-rlexp}.
    Note that the same construction applies when $g$ may contain $\nu$s.
    Now, for $f$ a regular expression, we can compute $f^\omega$ by $\nu X (f\bullet X)$, and then $ef^\omega$ similarly.\footnote{Note here that, in our semantics of both finite and infinite words, we are not using the fact that $\epsilon \notin \lang f$.}
    Since $\mu\nu$-expressions are closed under sums, they may express all $\omega$-regular languages wrt.~$\lang\cdot$.
\end{proof}

To prove the converse, we will need to first introduce a mechanism for deciding word membership.

\subsection{Evaluation puzzle}
As an engine for our main metalogical results about $\nu\CRLA$ we will first establish a characterisation of membership checking via games, or in this case puzzles.

The \defname{evaluation puzzle} is a puzzle (i.e.\ one-player game) whose states are pairs $(w,e)$ where $w\in \Alphabet^{\leq \omega}$ and $e$ is a $\mu\nu$-expression.
A \defname{play} of the puzzle runs according to the rules in \cref{fig:eval-rules}: puzzle-play is deterministic at each state except when the expression is a sum, in which case a choice must be made.
A play of the evaluation puzzle is \defname{winning} if:
\begin{itemize}
    \item it terminates at the winning state $(\epsilon,1)$; or,
    \item it is infinite and the smallest expression occurring infinitely often is a $\nu$ formula.
\end{itemize}

\begin{figure}
    \centering
    \begin{tabular}{|c|c|}
        \hline
        State  & Admissible moves
        \\\hline
       $(\epsilon, 1)$ & (Winning) \\
       $(aw,ae)$ & $(w,e)$ \\
       $(w,0)$ & - \\
       $(w,e+f)$ & $(w,e)$, $(w,f)$ \\
       $(w, \mu X e(X))$ & $(w,e(\mu X e(X))$ \\
       $(w, \nu X e(X))$ & $(w,e(\nu X e(X))$ 
       \\
        \hline
    \end{tabular}    
    \caption{Rules of the evaluation puzzle}
    \label{fig:eval-rules}
\end{figure}

\abhishek{maybe say no admissible moves here with $w$ (isntead of $\epsilon$ and then later say if you stop at 0 or at 1 with $w\ne\epsilon$ you lose and if you stop at 1 with $\epsilon$ you win.}

For the final point the `smallest expression occurring infinitely often' in a play is indeed well-defined. 
To see this, first extend the definition of $\fl\cdot$ from \cref{def:fl} to all $\mu\nu$-expressions by adding the clause:
\begin{itemize}
    \item if $\nu X f \in \fl e$ then $f(\nu X f(X)) \in \fl e$.
\end{itemize}
We extend all the associated order notation for $\FL$ to $\mu\nu$-expressions too, and we note that the properties of $\FL$ from \cref{prop:fl-props} still hold for our extended notions, by the same proof.

\abhishek{for lics, shall we introduce $\mu\nu$ expressions at one go and then restrict? Then exposition on FL,etc. will be uniform.}
\anupam{I think it is good to keep $\mu$-only stuff as warm up. stet for now.}

Now since any infinite play of the evaluation puzzle induces a $\leqfl$-decreasing trace, by the right components, there must be a smallest infinitely occurring expression that is either a $\mu$ or $\nu$ formula, by \cref{prop:fl-props}.

We have the following characterisation of evaluation:

\begin{theorem}
[Evaluation]
\label{thm:evaluation-theorem}
    $w \in \wlang e \iff $ there is a winning play from $(w,e)$.
\end{theorem}

The proof of this result uses relatively standard but involved techniques, requiring a detour through a theory of approximants and signatures when working with fixed point logics. 

First let us use this machinery to establish the converse of \cref{prop:omega-reg-have-munu-exps}, by showing that evaluation puzzles can be evaluated by non-deterministic parity automata.
To be clear, in our setting admitting both finite and $\omega$-words, we allow (non-deterministic) $\omega$-word automata to have $\epsilon$-transitions.\todo{find a citation? Otherwise justify or work modulo \cref{rem:fin-words-as-omega-words}. }\abhishek{To my surprise, couldn't find anything}
\anupam{Let's discuss just removing $1$ from $\nCRLA$ altogether.}
\begin{corollary}
    [$\omega$-regularity]
    \label{cor:munu-exprs-compute-omega-reg-lang}
    Every $\mu\nu$-expression computes an $\omega$-regular language.
\end{corollary}
\begin{proof}
    Given a $\mu\nu$-expression $e$
    construct a non-deterministic parity automaton by:
    \begin{itemize}
        \item the states are just $\fl e$;
        \item $e$ is the initial state;
        \item each moves of the evaluation puzzle (projected to the right component) is a transition, all being $\epsilon$ except $ae \to e$ which is an $a$-transition;
        \item the state $1$ loops on itself by an $\epsilon$-transition and has (arbitrary) even parity (even is accepting);
        \item other states are coloured according to the size of formula (smaller is better) ensuring that $\nu $-formulas are even and all others are odd.
    \end{itemize}
    Now the runs of this automaton on an input word $w$ correspond precisely to plays of the evaluation puzzle, with the parity condition of the former corresponding to the winning condition of the latter.
\end{proof}

\subsection{Approximants and assignments: proof of the Evaluation Theorem}
\anupam{this section can be appendix for the conference paper: we just need the evaluation puzzle result.}
The point of this subsection is to prove \cref{thm:evaluation-theorem}.
The reader comfortable with that result may safely skip straight to the next subsection.


Let us write $\subform$ for the subformula relation.
\begin{definition}
    [Dependency order]
Define the \defname{dependency order} $\dleq\, \df\, \leqfl \times \supform$, i.e.\ $e\dleq f$ if either $e\lefl f$ or $e\eqfl f $ and $f \subform e$.
\end{definition}
Note that, by the properties of FL closure, $\dleq$ is a well partial order on expressions. 
In the sequel we assume an arbitrary extension of $\dleq$ to a total well-order $\leq$.

\begin{definition}
    [Signatures]
    Let $M$ be a finite set of $\mu$-expressions $\{\mu X_0 e_0  \dge \cdots \dge \mu X_{n-1} e_{n-1}\}$.
    An $M$-\defname{signature} (or $M$-\defname{assignment}) is a sequence $\vec \alpha$ of ordinals indexed by $M$.
    Signatures are ordered by the lexicographical product order.
    An $M$-\defname{signed} formula is an expression $e^{\vec \alpha}$, where $e$ is a $\mu\nu$-expression and $\vec \alpha$ is an $M$-signature.

When $N$ is a finite set of $\nu$-formulas we define $N$-signatures similarly and use the notation $e_{\vec \alpha}$ for $N$-signed formulas.
\end{definition}

We evaluate signed formulas in $\lang\cdot$ just like usual formulas, adding the clauses,
\begin{itemize}
\item $\wlang {(\mu X_i e_i(X))^{\vec \alpha_i 0 \vec \alpha^i}} \df \emptyset$.
    \item $\wlang {(\mu X_i e_i(X))^{\vec \alpha_i (\alpha_i + 1 ) \vec \alpha^i}} \df \wlang { (e_i (\mu X_i e_i(X)))^{\vec \alpha_i \alpha_i  \vec \alpha^i}}$.
    \item $\wlang{(\mu X_i e_i(X))^{\vec \alpha_i \alpha_i \vec \alpha^i}} \df \bigcup \limits_{\beta_i <\alpha_i} \wlang{(\mu X_i e_i(X))^{\vec \alpha_i \beta_i \vec \alpha^i}} $, when $\alpha_i$ is a limit.

    \medskip
    
\item $\wlang {(\nu X_i e_i(X))_{\vec \alpha_i 0 \vec \alpha^i}} \df \Alphabet^{\leq \omega}$.
    \item $\wlang {(\nu X_i e_i(X))_{\vec \alpha_i (\alpha_i + 1 ) \vec \alpha^i}} \df \wlang { (e_i (\nu X_i e_i(X)))_{\vec \alpha_i \alpha_i  \vec \alpha^i}}$.
    \item $\wlang{(\nu X_i e_i(X))_{\vec \alpha_i \alpha_i \vec \alpha^i}} \df \bigcap \limits_{\beta_i <\alpha_i} \wlang{(\nu X_i e_i(X))_{\vec \alpha_i \beta_i \vec \alpha^i}} $, when $\alpha_i$ is a limit.
\end{itemize}
where we are writing $\vec \alpha_i \df (\alpha_j)_{j<i}$ and $\vec \alpha^i \df (\alpha_j)_{j>i}$.

Recall that least and greatest fixed points can be computed as limits of approximants. 
In particular, for any sets $M,N$ of $\mu,\nu$ formulas respectively, we have,
\begin{itemize}
    \item $\lang e = \bigcup \limits_{\vec \alpha} e^{\vec \alpha}$
    \item $\lang e = \bigcap \limits_{\vec \beta} e_{\vec \beta}$
\end{itemize}
where $\vec \alpha $ and $\vec \beta$ range over all $M$-signatures and $N$-signatures, respectively. 
Thus we have immediately:

\begin{proposition}
Suppose $e$ is an expression and $M,N$ the sets of $\mu,\nu$-formulas, respectively, in $\fl e$. We have:
\begin{itemize}
    \item If $w \in \wlang e$ then there is a least $M$-signature $\vec \alpha$ such that $w \in \wlang {e^{\vec \alpha}}$.
    \item If $w\notin \wlang e$ then there is a least $N$-signature $\vec \alpha$ such that $w \notin \wlang {e_{\vec \alpha}}$.
\end{itemize}
\end{proposition}

In fact it suffices to take only signatures of natural numbers, i.e.\ finite ordinals, for the result above, as we mention later in \cref{cor:closure-ordinals-omega}, but we shall not use this fact.
We are now ready to prove our characterisation of evaluation:

\begin{proof}
[Proof of \cref{thm:evaluation-theorem}]
Let $M,N$ be the sets of $\mu,\nu$-formulas, respectively, in $\fl e$. 

    $\implies$. 
    Suppose $w \in \wlang e $. 
    We construct a winning play from $(w,e)$, always preserving membership of the word in the language of the expression.
    In particular, at each state $(w',e_0 + e_1)$, we choose a summand $e_i$ admitting the least $M$-signature $\vec \alpha$ for which $w'\in \wlang {e_i^{\vec \alpha}}$.
    If the play is finite then it is winning by construction, so assume it is infinite, say $(w_i,e_i)_{i<\omega}$ and, for contradiction, that its smallest infinitely occurring formula is $\mu X f(X)$. 
    Write $\vec \alpha_i$ for the least $M$-signature s.t.\ $w_i\in \wlang {e_i^{\vec \alpha_i}}$, for all $i<\omega$.
    Then by construction $(\vec \alpha_i)_{i<\omega}$ is a monotone non-increasing sequence that, by progressiveness, does not converge. Contradiction.

    $\impliedby$. Suppose there is a winning play $\pi$ from $(w,e) $  and assume $w \notin \wlang e$ for contradiction. 
    By inspection of the puzzle rules, note that $\pi$ must be infinite, say $\pi = (w_i,e_i)_{i<\omega}$, and preserve non-membership, i.e.\ we always have $w_i \notin \wlang {e_i}$.
    Moreover, writing $\vec \alpha_i$ for the least $N$-signature $\vec \alpha_i $ such that $w_i \notin \wlang{(e_i)_{\vec \alpha_i}}$, we have by construction that $(\vec \alpha_i)_{i<\omega}$ is a monotone non-increasing sequence.
    Since $\pi$ is winning, this sequence cannot converge. Contradiction.
\end{proof}

Note that the asymmetry of the argument above, with no critical play choices for the $\impliedby$ direction, reflects the fact that we are working with a puzzle rather than a game.

\subsection{Cyclic system and soundness}
We work with the system $\nLRLAhat$ extending $\LRLAhat$ by the following rules:
\[
\vlinf{\lr \nu}{}{\nu X e(X) \seqar \Gamma}{e(\nu X e(X))\seqar \Gamma}
\qquad
\vlinf{\rr \nu}{}{e \seqar \Gamma, \nu X f(X)}{e \seqar \Gamma, f(\nu X f(X))}
\]

Once again we have a local soundness principle:
\begin{proposition}
    [Local soundness]
    \label{prop:local-soundness-nu}
    $\nLRLAhat \proves e \seqar \Gamma \implies \nRLA \proves e\leq \sum \Gamma$.
\end{proposition}\todo{give the extra cases in a proof here.}

Preproofs for this system are defined just as for $\LRLAhat$ before, but we need to be more nuanced about the definition of progress.

\begin{definition}
    [Traces]
    Fix a $\nLRLAhat$ preproof $P$ and an infinite branch $B = (\infrule_i)_{i<\omega}$, where each $\infrule_i$ is an inference step.  
    A \defname{$B$-trace} (or just `trace' when unambiguous) is an infinite sequence of formulas $(t_i)_{i<\omega}$, such that:
    \begin{itemize}
        \item Each $t_i$ occurs in the conclusion of $\infrule_i$.
        \item If $t_i$ is principal for $\infrule_i$, then $t_{i+1}$ is one of its auxiliary formulas.
        \item If $\infrule_i   $ is a $\kk a $ step, then $t_i = at_{i+1}$.
        \item Otherwise $t_i = t_{i+1}$.
    \end{itemize}
    An LHS or RHS trace is one that remains in the LHS or RHS respectively.
    We say that a trace $\tau$ is \defname{progressing} if it is infinitely often principal and:
    \begin{itemize}
        \item $\tau$ is LHS with smallest infinitely often principal formula a $\mu$ formula.
        \item $\tau $ is RHS with smallest infinitely often principal formula a $\nu$ formula.
    \end{itemize}
\end{definition}

Note that, by construction, any $B$-trace is indeed a trace in the sense of \cref{def:fl};
thus \cref{prop:fl-props} applies, and the notion of progress above is well-defined.

\begin{definition}
[System $\nCRLA$]
    A $\nLRLAhat$ preproof is progressing if each of its infinite branches has a progressing trace.
    We write $\nCRLA$ for the class of regular progressing $\nu\RLAhat$ coderivations.
\end{definition}

\todo{add some examples here}

The soundness argument now proceeds similarly to that of $\CRLA$, \cref{thm:crla-soundness}, only using the evaluation puzzle in place of NFAs.

\begin{theorem}
[Soundness of $\nCRLA$]
\label{thm:ncrla-soundness}
    If $\nCRLA \proves e \seqar \Gamma$ then $\wlang e \subseteq  \wlang \Gamma$.\todo{can strengthen this to progressing $\nLRLAhat$ preproofs.}
\end{theorem}
\begin{proof}  
Let $P$ be a $\nCRLA$ proof of $e\seqar \Gamma$ and suppose $w\in \wlang e$. We show $w\in\wlang \Gamma$. 
First, since $w\in \wlang e$ there is a winning play $\pi$ from $(w,e)$ by \cref{thm:evaluation-theorem}, which induces a unique (maximal) branch $B_\pi$ of $P$.
If $B_\pi$ is finite, then $w\in \lang \Gamma$ by the same argument as (the end of) the proof of \cref{thm:crla-soundness}.\anupam{maybe make that a lemma?}
Otherwise $w$ is infinite and $B_\pi$ must have a progressing trace.
Now, since $\pi$ is a winning play from $(w,e)$, the (unique) LHS trace of $B_\pi$ cannot be progressing, so it must have a RHS progressing trace $\tau$.
However, by construction of $B_\pi$, $\tau$ has a subsequence that forms (the right components of) a winning play from $(w,\sum \Gamma)$.
Thus indeed $w \in \lang \Gamma$ by \cref{thm:evaluation-theorem}.
\end{proof}

\section{Completeness of the guarded fragment via games}

Unlike the case of finite words, $\nCRLA$ is \emph{not} complete, in general, for $\wlang\cdot$. 

\begin{example}
    [Incompleteness]
    Consider the inequality $\top\leq 1 + \sum\limits_{a \in \Alphabet} a\top$, for $\top \df \nu X X$.
    This clearly holds in $\wlang \cdot$, as any word is either empty or begins with a letter but it is not provable in $\nCRLA$.
    For this note that the sequent $\top \seqar 1,\{a\top : a \in \Alphabet \}$ is a fixed point of bottom-up proof search (with no progressing thread).
    The same inequality is also not derivable in the aforementioned axiomatisation $\nRLA$. 
    For a countermodel simply interpret each letter $a\in \Alphabet$ as an appropriate constant function in, say, the complete lattice $\pow{\Alphabet^{\leq \omega}}$.
\end{example}

One potential remedy might be to expand the rules and syntax of $\nu\CRLA$ in a manner bespoke to the language model.
However, arguably more naturally, we rather show completeness for an adequate fragment of the syntax.

\subsection{The guarded fragment}
\label{sec:guarded-fragment}

We say that a variable occurrence $X$ in an expression $e$ is \defname{guarded} if $X$ occurs free in a subexpression of the form $af$.
We say that $e$ is guarded if all its variable occurrences are guarded.
A sequent is guarded if each of its expressions is.
It is not hard to see that the guarded expressions suffice to compute each regular language.
For this we rely on the well-known fact that regular languages can be factored into their identity and identity-free parts.

An \defname{identity-free regular expression} is generated by,
\[
e,f, \dots \quad \bnf \quad 
0 \quad \mid \quad a \quad \mid \quad e+f \quad \mid \quad ef \quad \mid \quad e^+
\]
and are interpreted as languages as expected, with $e^+$ computing the transitive closure of $e$. 
It is clear that $\epsilon \notin \lang e$ for any identity-free expression $e$.
In fact we have:

\begin{fact}
    For each regular expression $e$ there is $b_e \in \{0,1\}$ and $d_e$ identity-free such that $\lang e = \lang {b_e+d_e}$.
\end{fact}

In this way each $\omega$-regular language is computed by a sum of expressions of the form $ef^\omega$ and $f^\omega$, where $e$ and $f$ are identity-free.
It is not hard to see that the translation of \eqref{eq:prod-regexp-rlexp-to-rlexp} can be adapted to all such expressions with image in the guarded fragment.
Explicitly, for identity-free $e$ we define the right-linear $\nu$-expression $e\bullet g$ for all right-linear $\nu$-expressions $g$ by:
\[
\begin{array}{r@{\ \df\ }l}
     0 \bullet g & 0 \\
     a \bullet g & ag \\
    (e + f) \bullet g & e\bullet g + f \bullet g \\
    (ef) \bullet g & e \bullet (f \bullet g)  \\
    e^+\bullet g & \mu X (e\bullet  g + e\bullet X)
\end{array}
\]
It is not hard to see that $e \bullet g$ is always guarded, no matter what $g$ is, as long as $e$ is identity-free.
From here, for identity-free $e,f$, we can write $f^\omega$ as $\nu X (f\bullet X)$ and $ef^\omega$ as $e\bullet \nu X(f\bullet X)$.
This gives us:

\begin{proposition}
    If $A\subseteq\Alphabet^\omega$ is $\omega$-regular then there is a guarded $\mu\nu$-expression $e$ s.t.\ $\wlang e = A$.
\end{proposition}

From here the main result of this section is:

\begin{theorem}
    [Guarded completeness]
    \label{thm:ncrla-guarded-complete}
    If $e \seqar \Gamma$ is guarded and $\wlang e \subseteq \bigcup\limits_{f\in \Gamma} \wlang f$, then $\nu\CRLA \proves e \seqar \Gamma$.
\end{theorem}

\subsection{The proof search game}
In order to prove the above completeness result, we need to rely on certain game determinacy principles to organise bottom-up proof search appropriately.

\begin{definition}
    [Proof search game]
    The \emph{proof search game} (for $\nu\RLAhat$) is a two-player game played between Prover $(\prover)$, whose positions are inference steps of $\nu\RLAhat$, and Denier $(\denier)$, whose positions are sequents of $\nu\RLAhat$.
A \defname{play} of the game starts from a particular sequent:
at each turn, $\prover$ chooses an inference step with the current sequent as conclusion, and $\denier$ chooses a premiss of that step; the process repeats from this sequent as long as possible.

An infinite play of the game is \defname{won} by $\prover$ (aka \defname{lost} by $\denier$) if the branch constructed has a progressing trace; otherwise it is won by $\denier$ (aka lost by $\prover$). {In the case of deadlock, the player with no valid move loses.}
\end{definition}

Note that, as there are only finitely many possible sequents in a given preproof, the proof search game is finite state.
From here it is not hard to see that the set of progressing branches forms an $\omega$-regular language over the (finite) alphabet of possible sequents: guess a finite prefix and a progressing thread (along with its critical LHS-$\mu$ or RHS-$\nu$ formula), and verify that no smaller formula is unfolded after the prefix.\todo{refer to a result later or earlier.}
Consequently we have:

\begin{proposition}
[Essentially B\"uchi-Landweber]
    The proof search game from any sequent is finite memory determined.
\end{proposition}
Here `finite memory', in particular, means that the strategy needs only store a bounded amount of information at any time. 
It is not hard to see that any finite memory $\prover$ strategy is just a regular preproof (where the finite memory corresponds to multiple, yet finite, occurrences of the same sequent).
Moreover if the strategy is winning for $\prover$ then the corresponding preproof is progressing.
A similar analysis applies to $\denier$ strategies, and thus we have:\footnote{The fact that a player having a winning strategy means they have a \emph{regular} winning strategy can also be seen as a consequence of Rabin's basis theorem, as the set of strategies (for either player) forms an $\omega$-regular tree language.}
\begin{corollary}
    \label{cor:reg-prf-or-denier-win-strat}
    Any sequent is either $\nCRLA$-provable or has a $\denier$ winning strategy.
\end{corollary}

So, for completeness, it suffices to expand each $\denier$ winning strategy from $e\seqar \Gamma$ into a countermodel, i.e.\ a word in $\wlang e \setminus  \wlang \Gamma$.


\subsection{Proof of completeness}
Before giving our main completeness proof, we make an important observation about proof search in the guarded fragment.

Let us call an infinite play of the proof search game (or infinite branch of $\nCRLA$) \defname{fair} if it has no eventually stable traces.
A $\prover$ strategy (or $\nCRLA$ preproof) is fair if each of its infinite plays (branches) is.
It is not hard to see that:

\begin{proposition}
\label{prop:inf-plays-from-guarded-seqs-are-fair}
     Every infinite play from a guarded sequent is fair.
\end{proposition}
\begin{proof}
    Associate to each (guarded) expression $f$ the largest $\Alphabet$-free context $f^\Alphabet(\vec X)$ so that $f = f^\Alphabet(a_ig_i)_i$ for some $a_i \in \Alphabet$ and expressions $g_i$.
    For a sequent $S$ write $|S^\Alphabet|$ for the sum of the sizes of all $f^\Alphabet$s, for $f$ in $S$.
    Notice that each inference step except $\K$, bottom-up, reduces the $|\cdot^\Alphabet|$-value of the sequent.
    In particular this is the case for the fixed point rules by guardedness.
    Thus any infinite play must have infinitely many $\K$ steps, and thus have no eventually stable traces.
\end{proof}

We are now ready to prove our main completeness result:
\begin{proof}
[Proof of \cref{thm:ncrla-guarded-complete}]
By \cref{cor:reg-prf-or-denier-win-strat} assume WLoG there is a $\denier$ winning strategy $\delta$ from $e\seqar \Gamma$.
Let $B$ be an infinite play of $\delta$\anupam{why infinite? say more?}
%
%
and write $w_B\df \prod\limits_{\kk a \in B} a$, where the product ranges over all $\K$ steps in $B$, bottom-up.
Note that, by fairness, $w_B$ necessarily has length $\omega$.
We argue that $w_B \in \wlang e$ but $w_B \notin \wlang \Gamma$:
\begin{itemize}
    \item $w_B \in \wlang e$. By fairness the LHSs of $B$ induce an infinite play $\lambda_B$ of the evaluation puzzle from $(w_B,e)$, where the choices at sums are made according to $\delta$. 
    Now, as $\delta$ is $\denier$-winning, $B$ has no progressing traces, so the smallest infinitely often principal formula of its LHS trace must be a $\nu$ formula, which in turn must be the smallest expression occurring infinitely often along $\lambda_B$.
    Thus $\lambda_B$ forms a winning play from $(w_B,e)$, and so indeed $w_B \in \wlang e$ by \cref{thm:evaluation-theorem}.
    \item $w_B \notin \wlang \Gamma$.
    Any play $\pi$ of the evaluation puzzle from $(w_B,  f)$, for some $f\in \Gamma$, induces a RHS trace $\rho$ along $B$ that is infinitely often principal by fairness.
    Now, as $\delta$ is $\denier$-winning, $B$ has no progressing traces, so the smallest infinitely often principal formula of $\rho$ must be a $\mu$ formula, which in turn must the smallest expression occurring infinitely often along $\pi$.
    Thus $\pi$ is not a winning play from $(w_B,f)$.
    Since the choices of expression $f\in \Gamma$ and play $\pi$ from $(w_B,f)$ were arbitrary, we indeed have that $w_B \notin \wlang \Gamma$ by \cref{thm:evaluation-theorem}, as required.
    \qedhere
\end{itemize}
\end{proof}

\begin{remark}
    [`Uniqueness' of proofs]
    Note that \cref{prop:inf-plays-from-guarded-seqs-are-fair} essentially renders (pre)proofs of guarded sequents unique. 
    Since the logical rules are invertible they are also `confluent', and so any maximal (necessarily finite) derivation of logical rules with the same guarded end-sequent will have the same initial sequents, all of the form $ae\seqar a\Gamma$ by \cref{prop:inf-plays-from-guarded-seqs-are-fair}.
    In this sense we could have formulated our proof search game rather as a proof search \emph{puzzle}, from $\denier$'s point of view, essentially played against some/any (fair) $\prover$ strategy. 
    In the proof above this presented in the fact that we took an \emph{arbitrary} play $B$ of the $\denier$ winning strategy $\delta$.
\end{remark}

\section{Further work and discussion}
Before concluding this paper let us take some time to make further comments related to the work herein.


\todo{insert a subsection about complexity: proof checking and proof search.}

\subsection{Theories of greatest fixed points}\anupam{might be worth expanding this into a section in its own right.}
It would be natural to further investigate the extensions of $\RLA$ by greatest fixed points. 

On the one hand it would be natural to try show the completeness of (the guarded fragment) of $\nRLA$ for the model of $\omega$-regular languages $\wlang\cdot$. 
Indeed for this it seems promising to try adapt the techniques of \cref{sec:crla-to-rla} to the setting of greatest fixed points. 
Equational systems should now attribute a \emph{priority} to each variable, with acceptance given by a parity condition, similar to those for context-free $\omega$-grammars~\cite{Linna76,CG77,Niwinski84}, only restricted to right-linear syntax. 
From here the existence of appropriate extremal solutions to such systems requires a similar resolution of systems in $\RLA$, using Beki\'c's Theorem, \cref{lem:bekic}, only that the order of resolution should now respect the priority order.
Of course such development should be carried out in the guarded fragment, cf.~\cref{sec:guarded-fragment}, and in any case this is beyond the scope of this work.

On the other hand, adding the dual of $\mu$ motivates the consideration of the dual of $0$ and $+$ too, a native meet-semilattice structure $(\top,\cap)$. 
We might call the resulting class of algebras \emph{right-linear lattices}, in analogy with such extensions of Kleene algebras, Kleene lattices (see, e.g., \cite{kozen1994action,DasPous18:lka-pt}).
On the proof theoretic side such consideration seems natural, leading to genuinely two-sided sequents $\Gamma \seqar \Delta$ where the LHS cedent $\Gamma$ is interpreted as the \emph{meet} of its elements, similar to sequents of the usual classical sequent calculus.
Again this seems to generate a promising proof theory and, from a game theoretic point of view, meets turn our evaluation puzzles into evaluation \emph{games}, while the non-deterministic automata of \cref{cor:munu-exprs-compute-omega-reg-lang} become \emph{alternating}.
However let us point out that such a system is in fact \emph{unsound} for language models (or even relational models) as $a\top= \top$ becomes a theorem.
Proof theoretically this is due to the possibility of an \emph{empty} LHS; algebraically this corresponds to each letter being interpreted as (bounded) \emph{lattice} homomorphisms.
This is in contrast to the Kleene algebra setting where the extension by meets, Kleene lattices are still sound (but not complete) for language models.
Let us also point out that, in the case of \emph{multiple} formulas on the LHS, further problems arise for completeness too, e.g.\ because of $a\cap b = 0$ in the model $\lang\cdot$. 
For all these reasons we focussed strictly on co-intuitionistic sequents, with singleton LHS, throughout this work.

\subsection{Interpreting right-linear expressions as regular expressions}
Recalling \cref{sec:lhKAs}, each left-handed Kleene algebra can interpret $\mu$-expressions, due to the latter's capacity to resolve right-linear systems of equations (see, e.g., \cite{KozSil12:left-handed-completeness,KozSil20:left-handed-completeness}).
However $\mu$-expressions carry more information than equational systems, in particular an order of variable resolution, cf.~\cref{lem:bekic}.
Indeed this induces a somewhat simpler translation from $\mu$-expressions to regular expressions, equivalent in \emph{any} left-handed Kleene algebra:

   \begin{proposition}
    \label{prop:mu-exprs-to-reg-exprs}
        For $\mu$-expressions $e$ over variables $\vec X$ there are (closed) regular expressions $e^X$, for $X \in \vec X$, and $e^1$ s.t.\ every $\lhKA$ satisfies $ e =  {\sum\limits_{X \in \vec X} e^X  X + e^1}$.
    \end{proposition}
    \begin{proof}
        We proceed by induction on the structure of $e$.
        \begin{itemize}
        \item ($0$ already has the required form.)
            \item ($1$ already has the required form.)
            \item ($X$ already has the required form.)
            \item If $e= e_0 + e_1$ then we have,
            \[
            \begin{array}{r@{\ = \ }ll}
               e+f  & \sum\limits_{X\in \vec X}e_0^X X + e_0^1 
               +
               \sum\limits_{X\in\vec X}e_1^X X + e_1^1 & \text{by inductive hypothesis}
               \\
                 & \sum\limits_{X\in \vec X} (e_0^X + e_1^X)  X + (e_0^1 + e_1^1) & \text{by laws of $\lhKA$}
            \end{array}
            \]
            which has the required form.
            \smallskip
            \item If $e = \mu Y f(Y)$ then we have,
            \[
            \begin{array}{r@{\ = \ }ll}
               f(Y)  & \sum\limits_{X\in \vec X} f^X  X + f^Y  Y + f^1 & \text{by inductive hypothesis} \\
               \mu Y f(Y) & \mu Y\left(\sum\limits_{X\in \vec X} f^X \cdot X + f^Y \cdot Y + f^1\right)  & \text{by monotonicity of $\mu$} \\
               \noalign{\smallskip}
               & (f^Y)^*\left(\sum\limits_{X\in \vec X} f^X \cdot X + f^1\right) & \text{by laws of $\lhKA$} \\
               \noalign{\smallskip}
               & \sum\limits_{X\in \vec X} (f^Y)^*f^X \cdot X + (f^Y)^* f^1 & \text{by laws of $\lhKA$}
            \end{array}
            \]
            which has the required form. \qedhere
        \end{itemize}
    \end{proof}

Note that, for $\omega$-regular languages, we can extend the translation above to $\mu\nu$-expressions too, by exploiting the equation $\nu X (e + fX) = f^*e + f^\omega$ in $\lang \cdot$.
Thus, by reduction to the usual fixed points $*$ and $\omega$, we can deduce an (expected) property of closure ordinals for $\mu\nu$-expressions:\anupam{I think this holds in arbitrary $\nRLA$s too. E.g.\ in language model $f^\omega$ is either $0$ or $\top$ (if $f$ has empty word)}

\begin{corollary}
\label{cor:closure-ordinals-omega}
Each $\mu\nu$-expression $e(X)$ computes an operator $A\mapsto \lang {e(A)}$ whose least and greatest fixed points have closure ordinal $\leq \omega$.
\end{corollary}

\anupam{could say something here about interpretation of $\mu\nu$-expressions over languages of finite words still gives only regular languages.}


\subsection{What about context-free grammars and Chomsky Algebras?}
By extending $\mu$-expressions by arbitrary products, we have a well-known notation for context-free grammars that has also been studied from a logical point of view, in the guise of \emph{Chomsky algebras} and (standard) \emph{Park $\mu$-semirings} \cite{Leiss92:ka-with-rec,GraHenKoz13:inf-ax-cf-langs,EsikLeiss05:alg-comp-semirings,Leiss16:matrices-over-mu-cont-chom-alg}.
Such models can resolve even non-linear systems of equations, mimicking context-free grammars.
As universality of context-free languages is already $\Pi^0_1$-complete, there cannot be any recursive (or even recursively enumerable) axiomatisation for the equational theory of such $\mu$-expressions.
However Grathwohl, Henglein and Kozen give an infinitary axiomatisation in \cite{GraHenKoz13:inf-ax-cf-langs} by adding to the axiomatisation of idempotent semirings a notion of \emph{$\mu$-continuity}: $e \mu Xf(X) g = \sum\limits_{n<\omega} e f^n(0) g$.

Such continuity principles have been studied proof theoretically in the context of Kleene algebras, as well as extensions by meets and residuals, where they correspond to $\omega$-branching infinitary systems~\cite{Palka07}.
These in turn correspond to non-wellfounded proof systems whose proofs are not necessarily regular \cite{DasPous18:lka-pt}.
It would thus be interesting to investigate extensions of the hypersequential system $\HKA$ from \cite{DasPou17:hka} by least fixed points to attain a similar adequacy result.

Furthermore, extending $\HKA$ by least and \emph{greatest} fixed points gives a system for reasoning about $\omega$ context-free grammars.
Like in this work, completeness must may be driven by a game determinacy principle, but non-wellfoundedness means that the proof search game is no longer finite state: in general, \emph{analytic} determinacy is required, a principle beyond ZFC. 
This technique has nonetheless been employed in recent works on non-wellfounded proof theory~\cite{DDS23,DasGieMar23:igl-preprint}.
From here it would be interesting to apply \emph{projection} arguments translating non-wellfounded proofs to well-founded infinitely branching ones cf.~\cite{Studer08,DDS23}.
This would effectively yield an (infinite) axiomatisation of the theory of $\omega$ context-free languages over $\mu\nu$-expressions, extending the results of \cite{GraHenKoz13:inf-ax-cf-langs} to the infinite word setting.

\anupam{I commented further questions below}
\abhishek{I added a few questions there}


    
    
    
    




\section{Conclusion}
In this work we initiated an investigation of the proof theory of regular languages, over a syntax for right-linear grammars instead of regular expressions (cf.~\cite{DasPou17:hka,DasPous18:lka-pt,DasDouPou18:lka-completeness,HazKup22:transfin-HKA}). 
Using $\mu$-expressions to describe grammars, the corresponding algebraic structures, $\RLA$, are similar to (but coarser than) the Park $\mu$-semirings and Chomsky algebras for context-free grammars (cf.~\cite{GraHenKoz13:inf-ax-cf-langs,Leiss16:matrices-over-mu-cont-chom-alg,EsikLeiss05:alg-comp-semirings}).\todo{say something about naturality of the algebraic structure?}
We defined a cyclic proof system $\CRLA$ for our syntax, and showed its soundness and completeness for the model $\lang\cdot$ of regular languages.
Compared to previous approaches, our right-linear syntax significantly simplifies the underlying proof theory, requiring only co-intuitionistic sequents instead of the hypersequents from \cite{DasPou17:hka,DasDouPou18:lka-completeness}.

We showed that $\RLA$ is also complete for $\lang\cdot$, by extracting inductive invariants from $\CRLA$ proofs. 
This required significant bootstrapping of $\RLA$ and its capacity to define least solutions to (right-linear) equational systems.
Here the expressivity of $\mu$-expressions along with the simplicity of co-intuitionistic sequents, due to right-linearity, facilitated the quick obtention of this completeness result, as compared to previous approaches, e.g.~\cite{KozSil12:left-handed-completeness,KozSil20:left-handed-completeness,DasDouPou18:lka-completeness}.

We studied the extension of $\CRLA$ by \emph{greatest} fixed points, $\nCRLA$, and compared its theory to the extension of $\lang\cdot$ to $\omega$-regular languages.
Employing game theoretic techniques, we showed soundness of $\nCRLA$ for $\lang\cdot$, and finally completeness for the guarded fragment.
A distinguishing feature from the world of regular expresions is that $\omega$-languages already form a right-linear algebra, and so our proof theoretic treatment of both regular and $\omega$-regular languages is more uniform than existing approaches, cf.~\cite{CraLauStr15:omega-regular-algebras,HazKup22:transfin-HKA}.

\bibliographystyle{alpha}
\bibliography{biblio}

\end{document}